\documentclass[lettersize,journal]{IEEEtran}
\usepackage{amsthm}
\usepackage{amsmath}
\usepackage{amsmath,amsfonts}
\usepackage{algorithmic}
\usepackage{algorithm}
\usepackage{array}
\usepackage[caption=false,font=normalsize,labelfont=sf,textfont=sf]{subfig}
\usepackage{textcomp}
\usepackage{stfloats}
\usepackage{cleveref}
\usepackage{url}
\usepackage{verbatim}
\usepackage{graphicx}
\usepackage{cite}
\DeclareMathOperator*{\argmax}{arg\,max}
\begin{document}
\newtheorem{lemma}{Lemma}
\crefname{lemma}{Lemma}{Lemmas}
\newtheorem{theorem}{Theorem}
\crefname{theorem}{Theorem}{Theorems}

\title{RIS with Coupled Phase Shift and Amplitude: Capacity Maximization and Configuration Set Selection}

%RIS with Coupled Phase Shift and Amplitude: Capacity Maximization and Configuration Selection

%\pagenumbering{gobble}

% Capacity Maximization of RIS with Coupled Phase Shift and Amplitude
\author{Seyedkhashayar Hashemi, Masoud Ardakani, and Hai Jiang   
\thanks{The authors are with the Department of Electrical and Computer Engineering, University of Alberta, Edmonton, AB T6G 1H9, Canada (e-mail: \{seyedkha, ardakani, hai1\}@ualberta.ca).}
%\thanks{Part of this paper was submitted to IEEE ICC 2025, which presents part of the results in Section III.}
}

\maketitle

\begin{abstract}
%Reconfigurable intelligent surface (RIS) technology is one of the promising candidates for the next generation of wireless communications.
A reconfigurable intelligent surface (RIS) is a planar surface that can enhance the quality of communication by providing control over the communication environment. Reflection optimization is one of the pivotal challenges in RIS setups. While there has been lots of research regarding the reflection optimization of RIS, most works consider the independence of the phase shift and the amplitude of RIS reflection coefficients. 
In practice, the phase shift and the amplitude are coupled and according to a recent study, the relation between them can be described using a function. In our work, we consider a practical system model with coupled phase shift and amplitude. We develop an efficient method for achieving capacity maximization by finding the optimal reflection coefficients of the RIS elements. The complexity of our method is linear with the number of RIS elements and the number of discrete phase shifts. We also develop a method that optimally selects the configuration set of the system, where a configuration set means a discrete set of reflection coefficient choices that a RIS element can take.

\end{abstract}


\section{Introduction}
\IEEEPARstart{R}{econfigurable}
intelligent surface (RIS) technology is considered as one of the key enabling components for 6G wireless communications. 
By providing control over the communication environment, the RIS technology offers numerous benefits such as energy efficiency, better coverage, higher data rate, lower cost, and more reliable communications \cite{WuZhangTWC2019}. 

A RIS is a planar surface consisting of a number of small elements \cite{ShenXuICL2019}. The surfaces are reconfigurable, meaning that once a surface is deployed, the characteristics of its elements can be adjusted using a controller \cite{TanINFOCOM2018}. Each RIS element can be configured to modify the amplitude and the phase of its incident signal \cite{ChengLiTWC2021,LiGuoTCOM2022,KunduLi2022}. %The elements can be configured individually \cite{LiGuoTCOM2022} or in groups \cite{KunduLi2022}. 

%RIS technology can aid us in several applications including coverage extension by creating a virtual line-of-sight (LoS) channel when an obstacle blocks the direct channel between the transmitter and the receiver [?], co-channel interference mitigation [?], physical layer security [?], improving channel ranks by providing additional paths between the transmitter and the receiver [?], and simultaneous wireless information and power transfer (SWIPT) [?].

%While RIS has the potential to provide us with numerous benefits, the technology comes with several challenges [?]. For RIS to be considered the ideal choice for future wireless communication generations, these challenges must be resolved first. Some of the most important challenges are RIS hardware implementation choice [?], channel state information (CSI) acquisition [?], reflection optimization [?], and deployment position [?].  

RIS technology can aid us in coverage extension by generating a virtual line-of-sight (LoS) channel when an obstacle blocks the direct channel between the transmitter and the receiver \cite{BjornsonOzdoganWCL2020}. 
In such cases, appropriate adjustment of the reflection coefficients of the RIS elements is crucial. %Reflection optimization is the study of the optimal configuration of RIS elements \cite{HeShenWCL2022}. Suboptimal configuration of RIS elements leads to a lower data rate, higher power consumption, and inefficient usage of resources.
Reflection optimization plays an important part in the efficient usage of RIS and therefore is investigated in many recent studies \cite{LiuLiuTVT2023, WuZhangGLOBECOM2018, ZhangZhangJSAC2020, ZhangZhangISIT2020}.  %\cite{AbdallaICSCN2022, FengShenIOT2022, LiuLiuTVT2023}).

A commonly assumed scenario regarding RIS-aided communication systems is the multi-user setup, in which a base station communicates with several users with a RIS between them \cite{ZhengWuICL2020}. The problem is often formulated as a joint optimization of multiple variables \cite{LiCaiTCOM2021}. The alternating optimization (AO) technique \cite{TangMaICC2020} is often used to solve the problem. At each step of AO, one variable is optimized at a time while the rest are fixed \cite{SunWCL2023}. Several objectives can be considered for the optimization including sum-rate maximization \cite{PereraICC2022}, total transmit power minimization \cite{LiChenTWC2022}, energy efficiency maximization \cite{HuangTWC2019}, and maximization of minimum rate \cite{YangWCNC2020} or signal-to-interference-plus-noise ratio (SINR) \cite{NadeemTWC2020} for user fairness. 

On the other hand, there are ways to break the multi-user scenario into a number of single-user sub-problems. Methods such as RIS partitioning \cite{MakinArzykulovTWC2024} and distributed RIS deployment \cite{ZhangZhangTCOM2021} can help us achieve this goal. In RIS partitioning, the surface is divided into several segments, each serving a particular user \cite{KhaleelBasarJSTSP2022}. %The number of elements in each segment would be an additional variable that needs to be optimized \cite{KimKimTVT2022}.  
In a distributed deployment of RIS, instead of having one RIS serving multiple users, there will be multiple smaller surfaces each serving a particular user \cite{MaoYenerGLOBECOM2022}. Since each surface (or each surface partition) is now responsible for only one user, the setup can now be considered as multiple single-user communications. Therefore, many studies that consider a single-user setup have applications in more general multi-user setups too.

In the reflection optimization of a single-user setup, ideally, the amplitudes and the phase shifts of the reflection coefficients of the RIS elements are assumed to be continuously and independently adjustable \cite{WuTCOM2021}. In this scenario, the optimal solution can be achieved by aligning the RIS-aided paths with the direct path from the transmitter to the receiver while keeping the maximal possible amplitude of the reflection coefficients \cite{BjornsonWymeerschSPM2022}. However, achieving continuous adjustment for the phase shift is not possible in practice \cite{FaraRatajczak}. 
A more realistic assumption is to consider a finite number of discrete, evenly-spaced phase shifts \cite{PanJSTSP2022}. Unlike the continuous case, the optimization procedure is challenging for the discrete case. Several suboptimal approaches can be used such as quantizing the solution obtained from the continuous-phase-shift optimization problem \cite{LiuCST2021} or alternately optimizing the phase shifts \cite{WuZhangICASSP2019}. To obtain the globally optimal solution, methods such as exhaustive search and branch-and-bound (BB) \cite{WuZhangTCOM2020} can be used, but with high complexity. An efficient method is proposed in \cite{RenShenWCL2023} that can obtain the global optimality with linear complexity. 

While works in \cite{LiuCST2021,WuZhangICASSP2019,WuZhangTCOM2020,RenShenWCL2023} propose interesting solutions to determine RIS reflection coefficients with evenly-spaced discrete phase shifts, the assumption of having evenly-spaced phase shifts can be unrealistic in practical scenarios \cite{DaiAccess2020}. Hence the work in \cite{HashemiJiangTCOM2024} considers an arbitrary set of phase shifts of RIS reflection coefficients and achieves an optimal solution to determine RIS reflection coefficients in linear complexity.

%Considering the phase shifts to be evenly spaced may be infeasible due to hardware structure \cite{DaiAccess2020},  an algorithm with linear complexity has been developed in \cite{HashemiJiangTCOM2024} to achieve the global optimal for any arbitrary set of phase shifts.

All the aforementioned works consider that the phase shifts and amplitudes of RIS reflection coefficients can be independently adjusted. However, such an assumption is not feasible for a real RIS implementation \cite{ZhangICM2022}. In \cite{AbeywickramaZhangTCOM2020}, a practical model has been developed in which the amplitude is shown as a function of the phase shift, i.e., the amplitude and phase shift are coupled. The practical model has been used in several other works \cite{ZhaoPeng2024,ZhangZhangTVT2021}. In the literature, there has been no research that guarantees to achieve reflection optimization of RIS elements with the practical model. To fill this research gap, the following two major challenges should be addressed. 
\begin{itemize}
    \item Given a configuration set (here a {\it configuration set} is defined as a discrete set of reflection coefficient choices that a RIS element can take), how to optimally determine the RIS reflection coefficients for each channel realization of the system such that the maximal capacity is achieved? This challenge is referred to as {\it Capacity Maximization}.

\item How to select a configuration set for the system such that the average system capacity (averaged over all possible channel realizations) is maximized? This challenge is referred to as {\it Configuration Set Selection}.
    
\end{itemize}

We address both challenges in this paper. The contributions of this paper are summarized as follows 
 \begin{itemize}
     \item  Regarding capacity maximization with a given configuration set, we develop a method that yields the globally optimal RIS reflection coefficients that achieve capacity maximization. 
     The complexity of our method is linear with the number of RIS elements and linear with the size of the configuration set.
     \item To determine the optimal configuration set of the system, Monte Carlo simulations can be used, but with prohibitive complexity. To solve the problem in a much faster way, we theoretically prove that maximizing the average system capacity is approximately equivalent to maximizing the integral of a one-dimensional function. Thus, to get the optimal configuration set, we only need to find the configuration set in which the integral of the one-dimensional function is maximized. Our method is much faster than optimization based on Monte Carlo simulations, since for each configuration set, we only need to calculate an integral rather than running a large number of simulations. We also give a method to cut the running time of our method by almost half. 
      \end{itemize}

 %Additionally, we define a new challenge called ``Configuration Selection'' and propose a method for achieving the optimal solution. According to our numerical results, our methods provide apparent gains in terms of capacity.

The remainder of this paper is structured as follows. Section~\ref{sec:system_model} discusses the system model and the practical RIS model for coupled amplitude and phase shift of reflection coefficients. In Section \ref{sec:cap_max_section}, given a configuration set, our proposed method is presented to optimally solve the capacity maximization problem with linear complexity. In Section \ref{sec:conf_sel}, we present our method to optimally select a configuration set. Simulation results in Section \ref{sec:sim} show the performance of our proposed methods as well as comparison with other methods. Section \ref{sec:conclusion} concludes the work.

%commented only for conf
% The remainder of this paper is structured as follows. Section~II discusses the system model and the practical RIS configuration model. In Section III, our proposed method is presented to optimally solve the reflection optimization problem with polynomial complexity. Section IV introduces the configuration selection challenge and provides a method for solving it. Simulation results are presented and discussed in Section~V. The conclusions of this work can be found in Section VI.

\IEEEPARstart{}{} 
\section{System Model and Practical RIS Reflection Coefficient Model}\label{sec:system_model}
\subsection{System Model}
This paper considers that a transmitter communicates with a receiver as shown in Fig.~\ref{sysmod}. There is also a RIS between the transmitter and the receiver. The received signal at the receiver can be written as %$y[t] = h\cdot x[t] + w[t]$, 
\begin{equation}
    \label{received_signal_eq}
    y[t] = h\cdot x[t] + w[t],
\end{equation}
where \emph{$y[t] \in \mathbb{C}$} is the received signal, \emph{$x[t] \in \mathbb{C}$} is the transmitted signal, and $w[t] \sim \mathcal{N}_{\mathbb{C}}(0, N_0)$ is the additive white Gaussian noise (AWGN). 
%Assuming exactly one strong path (LOS) exists to and from RIS, all channels can be considered frequency flat 
The channel between the transmitter and the receiver, denoted as $h \in \mathbb{C}$, can be expressed as \cite{BjornsonWymeerschSPM2022}
\begin{equation}
    \label{Total_channel}
    h = h_0 + \sum_{n=1}^{N} h^{\prime}_n\theta_nh^{\prime\prime}_n.
\end{equation}
In \eqref{Total_channel}, $N$ is the number of RIS elements, \emph{$h_0 \in \mathbb{C}$} is the direct channel between the transmitter and the receiver, \emph{$h^{\prime}_n \in \mathbb{C}$} is the channel between the transmitter and the $n$th RIS element, \emph{$\theta_n=\beta_n e^{j\alpha_n} \in \mathbb{C}$} is the reflection coefficient of the $n$th RIS element, \emph{$h^{\prime\prime}_n \in \mathbb{C}$} is the channel between the $n$th RIS element and the receiver. 
Since \emph{h} is a complex number, it will have a magnitude and a phase. In this paper, $\angle x$ denotes the phase of the complex number $x$. 
\begin{figure}[!t]
\centering
\includegraphics[width=3.4in]{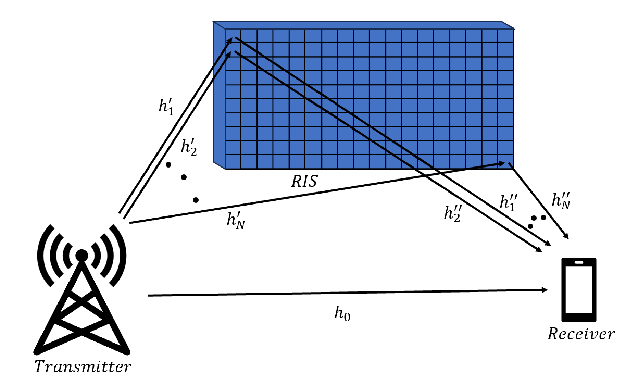}
\caption{The system model consisting of a transmitter, a receiver, and a RIS.}
\label{sysmod}
\end{figure}

The overall channel from the transmitter to the receiver through the $n$th RIS element can be expressed as
\begin{equation}
    \label{concat_channel_with_config}
    g_n = h^{\prime}_n\theta_nh^{\prime\prime}_n.
\end{equation}
It is also useful to define the cascaded channel coefficient $v_n \in \mathbb{C}$ for the $n$th RIS element as
\begin{equation}
    \label{concat_channel}
    v_n = h^{\prime}_nh^{\prime\prime}_n.
\end{equation}
\begin{figure}[!t]
\centering
\includegraphics[width=3.4in]{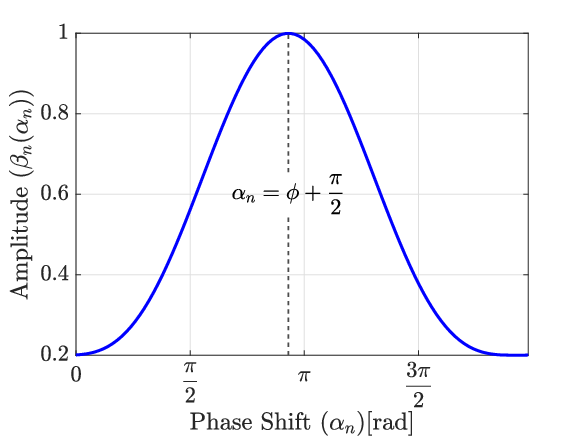}
\caption{Relationship between the phase shift and the amplitude.}
\label{amp_phas_relation}
\end{figure}
\begin{comment}
    
The received signal at the receiver can be written as:
\begin{equation}
    \label{received_signal_eq}
    y[t] = hx[t] + w[t]
\end{equation}
where \emph{$y[t] \in \mathbb{C}$} is the received signal, \emph{$x[t] \in \mathbb{C}$} is the transmitted signal, and $w[t] \sim \mathcal{N}_{\mathbb{C}}(0, N_0)$ is the additive white Gaussian noise (AWGN).
\end{comment}
The channel capacity from the transmitter to the receiver can be calculated as
\begin{equation}
    \label{capacity_eq}
    %C = B \log_2(1+ \text{SNR})= B\log_2(1+\frac{P|h|^2}{BN_0}) \quad \text{bits/s}.
C =  B\log_2(1+\frac{P|h|^2}{BN_0}) \quad \text{bits/s},
\end{equation}
in which $B$ is the transmitted signal bandwidth, $P$ is the transmitted signal power, $N_0$ is the noise power spectral density, and $\frac{P|h|^2}{BN_0}$ is the  signal-to-noise ratio (SNR).

\subsection{Practical RIS Reflection Coefficient Model}

In general, the reflection coefficient of the $n$th RIS element, denoted as $\theta_n$, is defined by two parameters, the amplitude ($\beta_n$) and the phase shift ($\alpha_n$). In other words, 
\begin{equation}
    \label{config_formula}
    \theta_n = \beta_n e^{j\alpha_n}.
\end{equation}

%\cite{YangZhangWCL2020, ZhangDaiTSP2021, BaiPanTWC2021}
In an ideal setup, the amplitude and the phase shift are independent and can take any possible value in a particular range
\begin{equation}
\beta_n \in [0,1], \alpha_n \in [0,2\pi).
\end{equation}
However, the ideal setup is not valid for practical RIS. According to \cite{AbeywickramaZhangTCOM2020}, for any circuit implementation of RIS, $\beta_n$ and $\alpha_n$ are not independent, and the relation between them can be expressed as
\begin{equation}
    \label{couple_formula}
    \beta_n(\alpha_n) = (1-\beta_{\min})\left(\frac{\sin(\alpha_n-\phi)+1}{2}\right)^\kappa + \beta_{\min},
\end{equation}
where $\beta_{\min}$, $\phi$, and $\kappa$ are all non-negative constants.
%In \cite{ZhangZhangTVT2021, ZhaoPengICL2024}, [?] authors define their system model based on \eqref{couple_formula}.
Such a practical model motivated us to consider the phase shift and the amplitude to be coupled rather than independent for our system model. For better illustration, Fig.~\ref{amp_phas_relation} demonstrates an example for the relationship between $\beta_n$ and $\alpha_n$.\footnote{According to \cite{AbeywickramaZhangTCOM2020}, we use $\beta_{\min} = 0.2$, $\kappa = 1.6$, and $\phi = 0.43\pi$ in the example in Fig.~\ref{amp_phas_relation}.}

Moreover, assuming a continuous adjustment for the phase shift is infeasible in practical setups. A better and more realistic assumption would be to consider a finite number of discrete choices of phase shift \cite{WuZhangTCOM2020}.
Therefore, in our setup, for each RIS element, say the $n$th element, $\theta_n$ is chosen from a {\it configuration set}, i.e., a finite set of $K$ choices:
\begin{equation}
    \label{configuration_set}
    \theta_n \in \{{\hat{\beta}_{1}e^{j\hat{\alpha}_{1}}}, {\hat{\beta}_{2}e^{j\hat{\alpha}_{2}}}, ..., {\hat{\beta}_{K}e^{j\hat{\alpha}_{K}}}\},
\end{equation}
in which amplitude $\hat{\beta}_{k}$ and phase shift $\hat{\alpha}_{k}$ are constants and satisfy (\ref{couple_formula}) for $k\in \{1,2,...,K\}$.

%In Section \ref{sec:cap_max}, for each channel realization, we will maximize the channel capacity given the $K$ choices of $\theta_n$ as shown in (\ref{configuration_set}). %Then in Section \ref{sec:conf_sel}, we will discus how to determine the $K$ choices. 

\section{Optimal Solution for Capacity Maximization}\label{sec:cap_max_section} %Given the $K$ Choices for Reflection Coefficient}

For a system with a given configuration set as shown in (\ref{configuration_set}), our goal is to maximize the capacity for each channel realization of the system by finding the optimal reflection coefficients of the RIS elements. As seen in \eqref{capacity_eq}, $|h|$ should be maximized to get the maximal capacity. Therefore, $\theta_n$ should be chosen in a way that the summation in \eqref{Total_channel} has the largest possible magnitude. The capacity maximization problem for any given channel realization of the system (i.e., given $h_0, v_1,v_2,...,v_N$), therefore, can be formulated as:
\begin{equation}
\label{probform}
\begin{aligned}
\max_{\theta_1,\theta_2...,\theta_N} \quad & |h|\\
\textrm{s.t.} \quad &    \theta_1,\theta_2...,\theta_N \in \{{\hat{\beta}_{1}e^{j\hat{\alpha}_{1}}}, {\hat{\beta}_{2}e^{j\hat{\alpha}_{2}}}, ..., {\hat{\beta}_{K}e^{j\hat{\alpha}_{K}}}\}.
\end{aligned}
\end{equation}

We will optimize the problem in \eqref{probform} in two steps. Consider $\theta_n^*$ as the optimal reflection coefficient for the $n$th RIS element and $h^*$ as the resulting optimal channel between the transmitter and the receiver. First, assuming $\angle h^*$ (i.e., the phase of $h^*$) is known, we will develop a method for determining $\theta_n^*$ for all elements. We will discuss this step in detail in Section III-A. 

In practice, $\angle h^*$ is not known at the beginning. Thus, in the next step, we have to go through all possibilities of $\angle h^*$ and find the one with the largest $|h|$. This may seem an impossible task since there will be infinite possibilities for $\angle h^*$. However, we will prove that by going through just a finite number of possibilities for $\angle h^*$, we will be able to find the optimal solution. Sections \ref{sec:inf_to_finite} $\sim$ \ref{sec:one_to_all_elements} give details of our method and related analysis, insights, and proofs. % some insight into how infinite possibilities for $\angle h^*$ can be converted into a finite number of possibilities. The number of possibilities of $\angle h^*$ that we have to go through is calculated in Section III-C. In Section III-D a method is provided for determining the possibilities of $\angle h^*$ that have to be checked for a single element. Section III-E describes how the results of Section III-D can be used to determine the possibilities of $\angle h^*$ for all elements. In Section III-F, an algorithm is proposed for calculating $\theta_n^*$ by checking all possibilities.

\begin{comment}
    
Next, we go through all the possibilities of $\angle h^*$. For each possibility, we calculate $\theta_n^*$ and $|h^*|$. In the end, we choose the possibility with the maximum $|h^*|$, and the $\theta_n^*$ of that possibility will be our optimal solution.
\end{comment}

\subsection{Determining $\theta_n^*$ by Assuming $\angle h^*$ is Known} \label{sec:cap_max}
Assume $\angle h^*$ is known. Consider the $n$th RIS element. According to \eqref{configuration_set}, there will be $K$ choices for $\theta_n$. Thus, there will also be $K$ choices for $g_n$ (expression of $g_n$ is given in (\ref{concat_channel_with_config})), i.e.,  
$g_n \in \{g_{n,1},g_{n,2},...,g_{n,K}\}$, with 
\begin{equation}\label{eq:g_nk_def}
g_{n,k} = v_n \hat{\beta}_k e^{j\hat{\alpha}_k}
\end{equation}
for $k=1,2,...,K$. 

Let us define $\langle h^*,g_{n,i}\rangle$ as:
\begin{equation}
    \label{inner_prod}
    \langle h^*,g_{n,i}\rangle = |h^*|\cdot |g_{n,i}|\cos(\angle h^*-\angle g_{n,i}).
\end{equation}
If we view complex numbers $h^*$ and $g_{n,i}$ as vectors in a complex plane, then $\langle h^*,g_{n,i}\rangle$ is actually the real inner product of vector $h^*$ and vector $g_{n,i}$. 

We have the following theorem.
\begin{theorem}
\label{theo1}
    If $\langle h^*,g_{n,i}\rangle$ is the maximum among $\{\langle h^*,g_{n,1}\rangle,\langle h^*,g_{n,2}\rangle,...,\langle h^*,g_{n,K}\rangle\}$, then the optimal reflection coefficient of the $n$th RIS element, denoted as $g_n^*$, is $g_{n,i}$.
\end{theorem}
    \begin{proof}
Please refer to Appendix A.

    \end{proof}
According to ~\cref{theo1}, among the $K$ choices for $g_n$ of the $n$th RIS element, all we have to do is to find the one that has the largest $\langle h^*,g_{n,i}\rangle$ for $i \in \{1,2,...,K\}$.

 Also, according to \eqref{eq:g_nk_def}, we have:
\begin{equation}
    \label{concat_channel_mag}
    |g_{n,i}| =  |v_n|\hat{\beta}_{i}.
\end{equation}
Using \eqref{concat_channel_mag}, the right-hand side of \eqref{inner_prod} can now be updated as:
\begin{equation}
    \label{inner_prod_updated}
    \langle h^*,g_{n,i}\rangle = |h^*|\cdot |v_n|\hat{\beta}_{i}\cos(\angle h^*-\angle g_{n,i})
\end{equation}
in which $\angle g_{n,i} = \angle (v_n \hat{\beta}_i e^{j\hat{\alpha}_i}) = \angle v_n+ \hat{\alpha}_i$ (from \eqref{eq:g_nk_def}).

As seen in \eqref{inner_prod_updated}, $|h^*|$ and $|v_n|$ are the same for all the members of $\{\langle h^*,g_{n,1}\rangle,\langle h^*,g_{n,2}\rangle,...,\langle h^*,g_{n,K}\rangle\}$. Thus, finding the maximum $\langle h^*,g_{n,i}\rangle$ would be equivalent to finding the largest $\hat{\beta}_i\cos(\angle h^*-\angle g_{n,i})$. So if $\angle h^*$ is known, we can determine the optimal reflection coefficient for each RIS element (say the $n$th RIS element) by finding the $g_{n,i}$ with the largest $\hat{\beta}_{i}\cos(\angle h^*-\angle g_{n,i})$. From here on, we refer to the largest $\hat{\beta}_{i}\cos(\angle h^*-\angle g_{n,i})$ as $\beta_{n}^*\cos(\angle h^*-\angle g_{n}^*)$ and the corresponding $g_{n,i}$ as $g_{n}^*$.

\subsection{Converting Infinite Possibilities of $\angle h^*$ to a Finite Number of Possibilities}\label{sec:inf_to_finite}

In Section \ref{sec:cap_max}, we find the optimal reflection coefficient for each RIS element with given known $\angle h^*$. However, in reality, $\angle h^*$ is unknown at the beginning. One method is to go through all possibilities of $\angle h^*$ to find the best one. %For each possibility, calculate $g_n^*$ for each element according to Section III-A. Store the values of $|h^*|$ and $\{g_1^* , g_2^* ,..., g_N^*\}$. At the end, find the largest $|h^*|$ and the corresponding $\{g_1^* , g_2^* ,..., g_N^*\}$ will give us the optimal solution.
Since there are infinite possibilities for $\angle h^*$, this method will not be feasible. But as we will soon show, there is a way to go through a finite number of possibilities for $\angle h^*$ and still be able to find the global optimal solution. This is possible because $g_n^*~ (n=1,2,...,N)$ keep unchanged over a range of $\angle h^*$. 

To get a better picture, Fig.~\ref{curves_examp} shows an example for the $n$th RIS element. In this example, there are $K=4$ choices for $\theta_n$, where the four choices are selected randomly from the curve in Fig.~\ref{amp_phas_relation}. The resulting $g_{n,1},g_{n,2},g_{n,3},g_{n,4}$ in this example are 1.7279, 3.0369, 4.0841, 5.6549 rad, respectively. Recall that we should find the largest one among $\hat{\beta}_{i}\cos(\angle h^*-\angle g_{n,i})$, $i\in\{1,2,3,4\}$. Fig.~\ref{curves_examp} shows how the four curves $\hat{\beta}_{i}\cos(\angle h^*-\angle g_{n,i})$  ($i=1,2,3,4$) change with $\angle h^*\in [0,2\pi)$. For presentation brevity, we call curve $\hat{\beta}_{i}\cos(\angle h^*-\angle g_{n,i})$ as ``curve $\angle g_{n,i}$" in the legend of Fig.~\ref{curves_examp} and in the following discussion.
\begin{comment}
    
$\angle v_n$ is set to $0.3\pi$.
\end{comment}

According to Section \ref{sec:cap_max}, in Fig.~\ref{curves_examp} we are looking for the maximal of the four curves $\angle g_{n,i}$  ($i=1,2,3,4$) at each $\angle h^*$ value. The maximal of the four curves is the red curve shown in Fig.~\ref{curves_examp}.  In Fig.~\ref{curves_examp}, $o_1,o_2,o_3,o_4,p_1,p_2,p_3,p_4$ are $\angle h^*$ values of the intersections of the curves. From Fig.~\ref{curves_examp}, when $\angle h^*$ changes from 0 to $p_1$, curve $\angle g_{n,4}$ is always above the other three curves, and thus, $g^*_n$ remains unchanged (keeps as $g_{n,4}$). Similarly, when $\angle h^*$ changes within $[p_1,p_2]$, $[p_2,p_3]$, $[p_3,p_4]$, or $[p_4,2\pi)$, $g^*_n$ remains as $g_{n,1}$, $g_{n,2}$, $g_{n,3}$, or $g_{n,4}$, respectively. 

We refer to the interval of $\angle h^*$ (within $[0,2\pi)$) over which one curve stays above all other curves as the {\it active interval} of the curve. Thus, in Fig.~\ref{curves_examp}, the active intervals of curves $\angle g_{n,1}$, $\angle g_{n,2}$, and $\angle g_{n,3}$ are $[p_1,p_2]$, $[p_2,p_3]$, and $[p_3,p_4]$, respectively, and the active interval of curve $\angle g_{n,4}$ is the union of two sub-intervals at the two sides of $[0,2\pi)$: $[p_4,2\pi) \cup [0,p_1]$. For presentation simplicity, we represent $[p_4,2\pi) \cup [0,p_1]$ as $[p_4,p_1]$. As a summary, we have the following {\bf definition} for an interval of $\angle h^*$ written as $[x_1,x_2]$ where $x_1,x_2 \in [0,2\pi)$: 
\begin{itemize}
\item if $x_1<x_2$, then the interval is a continuous interval from $x_1$ to $x_2$ (e.g., the active intervals of curves $\angle g_{n,1}$, $\angle g_{n,2}$, and $\angle g_{n,3}$); 
\item if  $x_1>x_2$, then the interval is the union of two continuous sub-intervals located at the two sides of $[0,2\pi)$, which is $[x_1,2\pi) \cup [0,x_2]$ (e.g., the active interval of $\angle g_{n,4}$). 
\end{itemize}
In either case, $x_1$ and $x_2$ are called the left and right boundary of interval $[x_1,x_2]$.

We call the right boundary ($\angle h^*$ value) of the active interval of a curve $\angle g_{n, i}$ as the {\it active intersection} of the curve. For example, $p_2$ is the active intersection for curve $\angle g_{n, 1}$, and $p_1$ is the active intersection for curve $\angle g_{n, 4}$.

%Please note that over the interval that one curve stays above all the others, $g^*_n$ remains unchanged. When the maximal curve changes, i.e., a new curve becomes the maximal curve, $g^*_n$ also changes. From now on, . For example, intervals $[p_1,p_2]$ or $[p_2,p_3]$ are two active intervals in Fig.~\ref{curves_examp}. In the first interval $g^*_n=g_{n, 1}$  because $\beta_{n, 1}\cos(\angle h^*-\angle g_{n, 1})$ is above other curves and on the second interval $g^*_n=g_{n, 2}$. 
%The right bound of an active interval is called an active intersection. For example, $p_2$ is the active intersection for $g_{n, 1}$.

\begin{comment}
    
$g_n^*$ stays the same over certain intervals. From now on, we will refer to these intervals as active intervals. Also, note that $g_n^*$ changes at certain intersection points. We will refer to the intersections as active intersections. For example, when $\angle h^*$ is in the interval between $p_1$ and $p_2$, $g_n^* = g_{n, 1}$. Thus $[p_1,p_2]$ will be an active interval for $g_{n,1}$. At $p_2$, $g_n^*$ switches from $g_{n, 1}$ to $g_{n, 2}$ as the configuration corresponding to the maximum of the curve changes. Thus, $p_2$ is an active intersection.
\end{comment}

\begin{figure}[!t]
\centering
\includegraphics[width=3.4in]{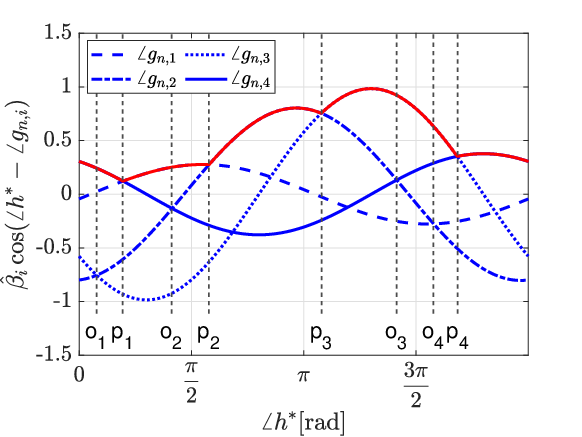}
\caption{$\hat{\beta}_i\cos(\angle h^*-\angle g_{n,i})$ versus $\angle h^*\in [0,2\pi)$ for the $n$th RIS element.}
\label{curves_examp}
\end{figure}
Since $g_n^*$ (and $\theta_n^*$) remains unchanged if $\angle h^*$ is within an active interval, we do not need to go through the infinite possibilities of $\angle h^*$. As long as we determine the active intervals, we can get $g_n^*$ (and $\theta_n^*$) in each active interval. For example, for Fig.~\ref{curves_examp}, when $\angle h^*$ is within active interval $[p_1,p_2]$, $[p_2,p_3]$, $[p_3,p_4]$, and $[p_4,p_1]$, $g_n^*$ is $g_{n,1}$, $g_{n,2}$, $g_{n,3}$, and $g_{n,4}$, respectively.

We will show how to determine the active intervals of one RIS element in Section \ref{sec:interval_one_element}, and extend the result to get the active intervals of all RIS elements in Section \ref{sec:interval_all_elements}.

%To proceed, first, we need to find out where the active intersection points are located.
 \subsection{Determining Active Intervals of One RIS Element}\label{sec:interval_one_element}
For each RIS element, say the $n$th element, its $g_n$ has $K$ choices: $g_{n,1},g_{n,2},...,g_{n,K}$, corresponding to $K$ curves called curve $\hat{\beta}_i\cos(\angle h^*-\angle g_{n,i})$ ($i=1,2,...,K$) over $\angle h^* \in [0,2\pi)$, similar to the curves in Fig.~\ref{curves_examp}. Recall that $\angle g_{n,i} = \angle v_n+ \hat{\alpha}_i$.

To find the intersections ($\angle h^*$ values) of curves $\hat{\beta}_i\cos(\angle h^*-\angle g_{n,i})$ and $\hat{\beta}_l\cos(\angle h^*-\angle g_{n,l})$, we need to solve the following equation:
     \begin{equation}
    \label{intersec_points_eq}
    \begin{aligned}
    &\hat{\beta}_i\cos(\angle h^*-\angle g_{n,i}) = \hat{\beta}_l\cos(\angle h^*-\angle g_{n,l}), \ i\neq l\\ 
    \begin{comment}
        
    \longrightarrow \quad& \beta_i(\cos(\angle h^*)\cos(\angle g_{n,i}) + \sin(\angle h^*)\sin(\angle g_{n,i})) \\
    =& \beta_l(\cos(\angle h^*)\cos(\angle g_{n,l}) + \sin(\angle h^*)\sin(\angle g_{n,l})) \\
    \longrightarrow \quad& \beta_i\cos(\angle h^*)\cos(\angle g_{n,i}) + \beta_i\sin(\angle h^*)\sin(\angle g_{n,i}) \\
    =& \beta_l\cos(\angle h^*)\cos(\angle g_{n,l}) + \beta_l\sin(\angle h^*)\sin(\angle g_{n,l}) \\
    \longrightarrow \quad& \beta_i\sin(\angle h^*)\sin(\angle g_{n,i}) - \beta_l\sin(\angle h^*)\sin(\angle g_{n,l})  \\
    =& \beta_l\cos(\angle h^*)\cos(\angle g_{n,l}) - \beta_i\cos(\angle h^*)\cos(\angle g_{n,i}) \\
    \longrightarrow \quad& (\beta_i\sin(\angle g_{n,i}) - \beta_l\sin(\angle g_{n,l}))\sin(\angle h^*)  \\
    =& (\beta_l\cos(\angle g_{n,l}) - \beta_i\cos(\angle g_{n,i}))\cos(\angle h^*) \\  
    \end{comment}
    \longrightarrow \quad& \tan(\angle h^*) =\frac{\hat{\beta}_l\cos(\angle g_{n,l}) - \hat{\beta}_i\cos(\angle g_{n,i})}{\hat{\beta}_i\sin(\angle g_{n,i}) - \hat{\beta}_l\sin(\angle g_{n,l})}. \\      
    \end{aligned}
    \end{equation}
From \eqref{intersec_points_eq}, the two curves intersect at two points in the range of $[0,2\pi)$: one at $\angle h^* = \arctan(\frac{\hat{\beta}_l\cos(\angle g_{n,l}) - \hat{\beta}_i\cos(\angle g_{n,i})}{\hat{\beta}_i\sin(\angle g_{n,i}) - \hat{\beta}_l\sin(\angle g_{n,l})}) \mod 2\pi$ and the other at $\angle h^* = (\pi + \arctan(\frac{\hat{\beta}_l\cos(\angle g_{n,l}) - \hat{\beta}_i\cos(\angle g_{n,i})}{\hat{\beta}_i\sin(\angle g_{n,i}) - \hat{\beta}_l\sin(\angle g_{n,l})})) \mod 2\pi$. Note that the two intersections are $\pi$ radians apart. 

For the $K$ curves of the $n$th RIS element, the total number of intersections will be ${2}{{K}\choose{2}}$. Among all the intersections, only the active intersections matter to us. For example, as seen in Fig.~\ref{curves_examp}, intersections $o_1$ and $o_2$ are not active and therefore not of any importance to our optimization goal.
 
Let us consider the intersections on a single curve. It has two intersections with any of the other $K-1$ curves. Thus, the total number of intersections on a single curve will be $2(K-1)$. For example, in Fig.~\ref{curves_examp}, since there are 4 curves in total, each curve has $2(4-1) = 6$ intersections on it. 

Considering the $2(K-1)$ intersections on each curve, there will be an interval between every two consecutive intersections, making the total number of intervals on each curve being $2(K-1)$.\footnote{The $2(K-1)$ intersections partition range $[0,2\pi)$ of $\angle h^*$ into $2(K-1)+1$ intervals. As aforementioned, the union of the interval from $0$ to the left-most intersection and the interval from the right-most intersection to $2\pi$ is viewed as a single interval.} Note that by an ``interval", we mean an interval of $\angle h^*$.

\begin{comment}  
Let us call the interval along which the considered curve is maximum among all curves an active interval.
\end{comment} 
In the following theorem, we prove that the considered curve can have at most one active interval among all $2(K-1)$ intervals.

\begin{theorem}
\label{theo2}
For each curve of the $n$th RIS element, at most one of the $2(K-1)$ intervals is active.
\end{theorem} 
    \begin{proof}
Please refer to Appendix B.

\end{proof}
According to ~\cref{theo2}, each curve will have at most one active interval. Recall that there are $K$ curves associated with each RIS element. Therefore, there will be at most $K$ active intervals for each RIS element. 
\begin{comment}
The intervals are connected, meaning that the right bound of the active interval of a particular curve is the left bound of the active interval of another curve. 
\end{comment}
This means that the number of active intersections for each RIS element can at most be $K$. Next, we will develop a method for finding the active intersections of each RIS element, say the $n$th RIS element.

\begin{comment}
\begin{figure}[!t]
\centering
\includegraphics[width=3.5in]{theo2_intervs.png}
\caption{Demonstration of active and inactive intervals}
\label{theo2_intervs}
\end{figure}
\end{comment}

%\subsection{Determining the active intersections of a single element}
%In this section, we propose two methods for calculating the active intersections of a single element.
%\subsubsection{Method 1 (Common range based)}
For the $n$th RIS element, it has $K$ curves. Suppose $I_{n,i,l}$ is the interval at which curve $\hat{\beta}_i\cos(\angle h^*-\angle g_{n,i})$ is above curve $\hat{\beta}_l\cos(\angle h^*-\angle g_{n,l})$, i.e., 
$\hat{\beta}_i\cos(\angle h^*-\angle g_{n,i}) > \hat{\beta}_l\cos(\angle h^*-\angle g_{n,l})$. $I_{n,i,l}$ can be determined with the help of the intersection points of the two curves as discussed at the beginning of Section \ref{sec:interval_one_element}. For the example in Fig.~\ref{curves_examp}, we have $I_{n,2,3} = [o_1,p_3]$ and $I_{n,3,2} = [p_3,o_1]$. 

For curve $\hat{\beta}_i\cos(\angle h^*-\angle g_{n,i})$, its active interval denoted as $I_{n, i}$ can be determined as 
     \begin{equation}
    \label{active_interv_eq}
    \begin{aligned}
    I_{n, i} = \bigcap_{\substack{l=1 \\ l \neq i}}^K I_{n , i, l}.
    \end{aligned}
    \end{equation}
Note that $\bigcap$ refers to the common range of intervals. The following function helps us find the common range of two intervals.

\begin{comment}
    
The key aspect of this method is a function called Common Range Calculator (CRC). First, we discuss this function and then we explain how this function is used to find the active intersection.
\end{comment}

\emph{Common range calculator (CRC):}
 The CRC function takes two intervals $[l_1,r_1]$ and $[l_2,r_2]$ as input intervals, and gets interval $[l_3,r_3]$ as the common range (CR) of the two input intervals. Here the intervals follow the interval {\bf definition} in Section \ref{sec:inf_to_finite}, and `$l$' and `$r$' mean the left and right boundary of an interval, respectively.  Table \ref{CRC_table} summarizes the CRC calculation results for different cases (in which ``No CR'' means that the two input intervals do not have a range in common), and Fig.~\ref{table_examples} provides a demonstration for each case in Table \ref{CRC_table}.

 %the left and right bounds of two intervals ($l_1, r_1, l_2, r_2 \in [0,2\pi)$) as input arguments, calculates the range in common between the two provided intervals, and outputs the left and right bounds of the calculated range ($l_3, r_3\in [0,2\pi)$). Based on the relation between the left and right bounds, we can determine the common range according to Table \ref{CRC_table}. ``No CR'' means the two intervals do not have a range in common. To get a better picture, Fig.~\ref{table_examples} provides an example for each case considered in Table \ref{CRC_table}.
 
By using the CRC, we can get $I_{n, 1}, I_{n, 2}, ..., I_{n, K}$, i.e., the $K$ active intervals of the $K$ curves associated with the $n$th RIS element. Based on the active intervals, we can get the active intersections for the active intervals of the $n$th RIS element.% (recalling that the right boundary of an active interval is called its active intersection). 

\begin{table}[!t]
\caption{CRC Calculation Results \label{CRC_table}}
\centering
\begin{tabular}{|c|c|c|c|}
\hline
Index &Case & $l_3$ & $r_3$\\
\hline
$I$&$(l_1<r_1) \& (l_2<r_2)$& &\\ 
&$\& \max(l_1,l_2) < \min(r_1,r_2)$ & $\max(l_1,l_2)$ &$\min(r_1,r_2)$\\
\hline
$II$&$(l_1<r_1) \& (l_2<r_2)$& &\\ 
&$\& \max(l_1,l_2) > \min(r_1,r_2)$ & No CR & No CR\\
\hline
$III$&$(l_1>r_1) \& (l_2>r_2)$& $\max(l_1,l_2)$ &$\min(r_1,r_2)$\\ 
\hline
$IV$&$(l_1>r_1) \& (l_2<r_2)$& &\\ 
&$\& (r_2<l_1) \& (r_1<l_2)$ & No CR & No CR\\
\hline
$V$&$(l_1>r_1) \& (l_2<r_2)$& &\\ 
&$\& (r_2>l_1)$ & $\max(l_1,l_2)$ & $\max(r_1,r_2)$\\
\hline
$VI$&$(l_1>r_1) \& (l_2<r_2)$& &\\ 
&$\& (l_2<r_1)$ & $\min(l_1,l_2)$ & $\min(r_1,r_2)$\\
\hline

$VII$&$(l_1<r_1) \& (l_2>r_2)$& &\\ 
&$\& (r_2<l_1) \& (r_1<l_2)$ & No CR & No CR\\
\hline
$VIII$&$(l_1<r_1) \& (l_2>r_2)$& &\\ 
&$\& (r_2>l_1)$ & $\min(l_1,l_2)$ & $\min(r_1,r_2)$\\
\hline
$IX$&$(l_1<r_1) \& (l_2>r_2)$& &\\ 
&$\& (l_2<r_1)$ & $\max(l_1,l_2)$ & $\max(r_1,r_2)$\\
\hline

\end{tabular}
\end{table}
    
\begin{figure}[!t]

\centering
\includegraphics[width=3.4in]{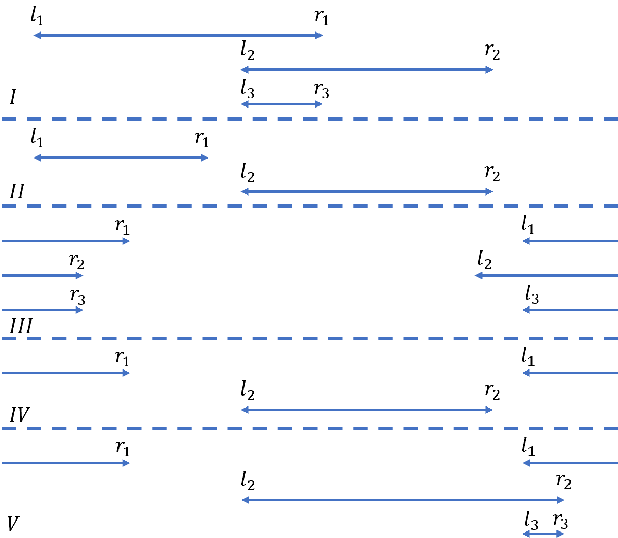}\\
\includegraphics[width=3.4in]{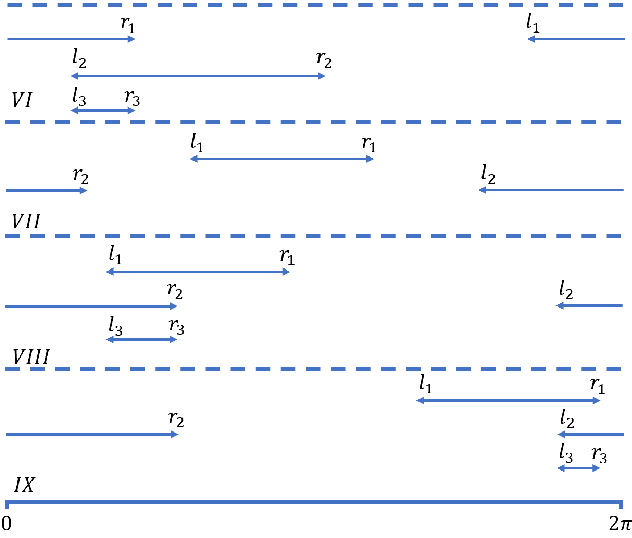}\\
\caption{Demonstration for the nine cases in Table \ref{CRC_table}.}
\label{table_examples}
\end{figure}

\subsection{Determining the Active Intersections of All Elements} \label{sec:interval_all_elements}
%According to the methods described in Section III-D, we can calculate all active intersections for a single element which is a list with a length of at most $K$. Let us discuss how we can calculate the active intersections of all elements just by knowing the active intersections of a single element.

After we get the active intersections for the $n$th RIS element, we can quickly get active intersections for any other RIS element, say the $m$th RIS element, as follows. 

%According to \eqref{concat_channel_with_config} and \eqref{configuration_set}, we have:
%\begin{equation}
%    \label{concat_channel_ang}
%   \angle g_{n,i} = \angle (v_n \hat{\beta}_i e^{j\hat{\alpha}_i}) = \angle v_n+ \hat{\alpha}_i.
%\end{equation}

Since $ \angle g_{n,i} =  \angle v_n+ \hat{\alpha}_i$, we have:
\begin{equation}
    \label{alt_curve_formula}
    \hat{\beta}_i\cos(\angle h^*-\angle g_{n,i}) = \hat{\beta}_i\cos(\angle h^*-\angle v_n - \hat{\alpha}_i).
\end{equation}

So $\hat{\beta}_i\cos(\angle h^*-\angle v_n - \hat{\alpha}_i)$ is the $i$th curve of the $n$th RIS element. If we shift this curve to the left by $\angle v_n - \angle v_{m}$ (i.e., we replace $\angle h^*$ with $\angle h^* + \angle v_n - \angle v_{m}$), we will get curve $\hat{\beta}_i\cos(\angle h^* + \angle v_n - \angle v_{m}-\angle v_n - \hat{\alpha}_i)$ = $\hat{\beta}_i\cos(\angle h^* - \angle v_{m} - \hat{\alpha}_i)$ = $\hat{\beta}_i\cos(\angle h^*-\angle g_{m,i})$, which is the $i$th curve of the $m$th RIS element. Therefore, if we shift all the curves of the $n$th RIS element by the constant value of $\angle v_n - \angle v_{m}$, we will end up with the curves of the $m$th RIS element. As a result, the active intersections of the $m$th RIS element would be the active intersections of the $n$th RIS element shifted by $\angle v_n - \angle v_{m}$. Thus, simply by shifting, we can calculate the active intersections of the rest of the elements.

\subsection{Determining Overall Optimal Reflection Coefficients}\label{sec:one_to_all_elements}
% In this section, we will describe a method that uses the active intersections to determine $\theta_n^*$ by checking a finite number of possibilities of $\angle h^*$.

\begin{comment}
    
If we repeat the procedure for all elements, there will be $N$ lists with a length of at most $K$. 
\end{comment}
In Section \ref{sec:interval_all_elements}, we have determined the active intersections for all RIS elements. Since each RIS element has at most $K$ active intervals and at most $K$ active intersections, there will be at most $NK$ active intersections in total for all $N$ RIS elements. The active intersections partition the range $[0,2\pi)$ of $\angle h^*$ into at most $NK$ regions. As discussed in Section \ref{sec:inf_to_finite}, when $\angle h^*$ is within one region, $\theta_1^*,\theta_2^*,...,\theta_N^*$ remain unchanged. Thus, we only need to go through at most $NK$ regions of $\angle h^*$, find $\theta_1^*,\theta_2^*,...,\theta_N^*$ and the corresponding channel capacity $C$ when $\angle h^*$ is within each region, and pick up the largest channel capacity and select the $\theta_1^*,\theta_2^*,...,\theta_N^*$ in the corresponding region as the overall optimal reflection coefficients for the $N$ RIS elements. Since we only need to go through at most $NK$ regions, our method has a complexity linear with $N$ and $K$.

\section{Configuration Set Selection}\label{sec:conf_sel}
 In Section \ref{sec:cap_max_section}, we have demonstrated how capacity maximization is performed for any channel realization of the system, assuming a given configuration set, i.e., a set of $K$ reflection coefficient choices $\{{\hat{\beta}_{1}e^{j\hat{\alpha}_{1}}}, {\hat{\beta}_{2}e^{j\hat{\alpha}_{2}}}, ..., {\hat{\beta}_{K}e^{j\hat{\alpha}_{K}}}\}$ as shown in (\ref{configuration_set}). In this section, our target is to select the optimal configuration set for the considered system such that the expected capacity is maximized.

 \begin{comment}
    
The set of choices is defined when the surface is being designed and manufactured and can not be changed once the surface is deployed in a communication network.
\end{comment}
%The question that we address in this section is on selecting the optimal configuration. Please note that this is completely different from our efforts in the previous section. Previously, we had the configuration as given and we decided the reflection coefficient at each RIS element to maximize the capacity. Here, we are not optimizing the RIS reflection coefficients. Instead, we aim to pick the configuration that the RIS optimization algorithm can select from. 

In the literature, most of the works assume that the phase shifts of the reflection coefficient choices are evenly spaced \cite{LiuCST2021,WuZhangICASSP2019,WuZhangTCOM2020,RenShenWCL2023}, i.e., $\{\hat{\alpha}_1,\hat{\alpha}_2,...,\hat{\alpha}_K\}= \{0, \frac{2\pi}{K},\frac{4\pi}{K} ..., \frac{2\pi(K-1)}{K}\}$. This setting is reasonable when the amplitude and phase shift of a reflection coefficient can be adjusted independently. However, in our system, the amplitude and phase shift are coupled. Thus, in general, a configuration set with evenly spaced phase shifts of the reflection coefficients does not guarantee optimality. %constant and not a function of the phase shift. However, as we can see in Fig.~\ref{amp_phas_relation}, different phase shifts may impose different amplitudes.
%As a result, an inappropriate configuration might lead to a low achievable capacity. Therefore, it is crucial to determine a suitable configuration for our surface elements.

\begin{comment}
    
when the surface is being designed.
\end{comment}
%Configuration selection is an offline optimization process that can be done by the RIS manufacturer. The goal is to provide a configuration that leads to the highest average capacity over all possible channel realizations. 
\begin{comment}
    
\begin{equation}
\label{probform2}
\begin{aligned}
\max_{\psi_k} \quad & \mathbb{E}(C) = \\
\textrm{s.t.} \quad &    \theta_n \in \{\psi_{n,1}, \psi_{n,2}, ..., \psi_{n,K}\}\\
& (\angle \psi_k, |\psi_k|) \in \{(\alpha_n,\beta_n(\alpha_n)):\alpha_n \in [0,2\pi) \}\\
\end{aligned}
\end{equation}
\end{comment}

\begin{comment}
\begin{figure}[!t]
\centering
\includegraphics[width=3.4in]{sysblock_rev1.png}
\caption{Block diagram demonstrating how configuration selection affects capacity maximization}
\label{sysblock}
\end{figure}
\end{comment}

Since amplitude is a function of phase shift, selecting a configuration set is equivalent to selecting $K$ phase shifts: $\hat{\alpha}_1,\hat{\alpha}_2,...,\hat{\alpha}_K$. 
As we can see in Fig.~\ref{amp_phas_relation}, the phase shifts can be any value within $[0,2\pi)$. Thus, theoretically, there are an infinite number of possible configuration sets for the system. To make the setup feasible, we consider selecting $K$ phase shifts from a large number, denoted as $M$ ($M\gg K$), of evenly spaced phase shifts over the range $[0,2\pi)$.
%\begin{equation}
%\Omega = \{0, \frac{2\pi}{M},\frac{4\pi}{M} ..., \frac{2\pi(M-1)}{M}\}.
%\end{equation}
We have the following notation definitions.
\begin{itemize}
\item Denote the set of $M$ evenly spaced phase shifts as $\Omega$.
\item Define a {\it $K$-size subset of $\Omega$} as a subset of $\Omega$ with the size of the subset being $K$. 
\item Define $\Phi$ as the set of all $K$-size subsets of $\Omega$. 
\end{itemize}
So we have $|\Phi| = {M \choose K}$. Therefore, our objective is to select a $K$-size subset of $\Omega$, denoted $\Psi$, such that the expected capacity of the system is maximized. In the sequel, $\Psi$ is also called a configuration set of the system.

%the following approximation of function $\beta_n(\alpha_n)$
%\begin{equation*}
%\label{approx_function}
%{\bar{\beta}_n(\alpha_n)} = \begin{cases}
%\beta_n(\alpha_n),&{\text{for}}\ {\alpha_n = \frac{2m\pi}{M}, } \\ &m\in\{0,1,...,M-1\}\\ \\
%{undefined,}&{\text{otherwise.}} 
%\end{cases}
%\end{equation*}
%The configuration $\Psi$ can be defined as
%\begin{equation}
%\label{config_possib}
%\begin{aligned}
%\Psi = &\{\psi_1, \psi_2, ..., \psi_K\},\\
%& (\angle \psi_k, |\psi_k|) \in \{(\alpha_n,\bar{\beta}_n(\alpha_n)):\\
%&\alpha_n = \frac{2m\pi}{M}, m\in\{0,1,...,M-1\} \}.\\
%\end{aligned}
%\end{equation}
The configuration set selection problem can be formulated as
\begin{equation}
\label{probform2_ref}
\begin{aligned}
\Psi^* = \argmax_{\Psi \in \Phi} \quad & \mathbb{E}(C_{\Psi}^*),% = \lim_{R\rightarrow\infty}\frac{1}{R}\sum_{r=1}^R C_r^*,\\
%\textrm{s.t.} \quad %&    \theta_n \in \{\psi_{n,1}, \psi_{n,2}, ..., \psi_{n,K}\}\\
%& (\angle \psi_k, |\psi_k|) \in \{(\alpha_n,\hat{\beta}_n(\alpha_n)):\\&\alpha_n = \frac{2m\pi}{M}, m\in\{0,1,...,M-1\} \}\\
\end{aligned}
\end{equation}
where $C_{\Psi}^*$ is the maximal capacity over a channel realization of the system with the configuration set $\Psi$, and the expectation $\mathbb{E}(\cdot)$ is over all channel realizations of the system. %expected value of the optimal capacity $C^*$. The right-hand side of  \eqref{probform2_ref} is derived based on the law of large numbers (LLN) where $R$ is the number of realizations.

%$\psi_k$ is the $k$th member of the selected set of $K$ configurations.

%Assuming the approximated function consists of $M$ points, there will be ${M}\choose{K}$ ways to select $\Psi$. The $\Psi$ resulting in the highest average capacity will be the optimal solution according to \eqref{probform2_ref}. 

%\subsection{Monte Carlo Simulation based (MSCB) Method for Configuration Selection}

An intuitive method to solve the configuration set selection problem in (\ref{probform2_ref}) is to use Monte Carlo simulations, referred to as {\it Monte Carlo Simulation based (MCSB) method}.
 In this method, we go through all ${M}\choose{K}$ options of $\Psi$. For each option, we simulate a large number, denoted $R$, of channel realizations of the system. For each realization, we calculate the maximum capacity using the method described in Section \ref{sec:cap_max_section}. Then for the option, we average the achievable maximal capacity associated over all $R$ channel realizations. In the end, we select the option with the largest average achievable capacity. 

 In MCSB method, a large number of channel realizations are required for each of the ${M}\choose{K}$ options of $\Psi$. Moreover, at each channel realization, the capacity maximization method in Section \ref{sec:cap_max_section} must be applied which has a complexity of $\mathcal{O}(NK)$. Thus, the total complexity of the MCSB method is $\mathcal{O}({{M}\choose{K}}RNK)$, which is time-consuming. Next, we will propose a much faster method that does not require any Monte Carlo simulations and has insights.
\begin{comment}
    
\begin{figure}[!t]

\centering
\includegraphics[width=2in]{approx12.png}\\

\includegraphics[width=2in]{approx24.png}\\
\includegraphics[width=2in]{approx36.png}
\caption{Demonstration of $\hat{\beta}(\alpha)$ for $M=12, M=24,$ and $M=36$}
\label{approx_examples}
\end{figure}
\end{comment}

\subsection{Integral Maximization Based (IMB) Configuration Set Selection}\label{sec:IMB}
Before we present our method, we introduce two functions $F_n(\angle h^*)$ and $S(\angle h^*)$ as follows.

In Section \ref{sec:cap_max_section}, for the $n$th RIS element, we have shown that we should get the maximal of curves $\hat{\beta}_{i}\cos(\angle h^*-\angle g_{n,i})$, $i\in\{1,2,...,K\}$ for any $\angle h^* \in [0, 2\pi)$ . Accordingly, for the $n$th RIS element, we can define $F_n(\angle h^*)$ as the maximal of the curves for a given $\angle h^*$ value, as
\begin{equation}
F_n(\angle h^*) = \max_{i\in\{1,2,...,K\}} \  \hat{\beta}_i\cos(\angle h^*-\angle g_{n,i}).
\end{equation}
Fig.~\ref{Fn} shows an example of function $F_n(\angle h^*)$ with $K=4$. For an RIS with $N$ elements, we have $N$ different functions: $F_1(\angle h^*),F_2(\angle h^*),...,F_N(\angle h^*)$.

Also define function $S(\angle h^*)$ as
\begin{equation}\label{S_func}
S(\angle h^*) = \max_{i\in\{1,2,...,K\}} \  \hat{\beta}_i\cos(\angle h^*-\hat{\alpha}_i).
\end{equation}
Accordingly, we have
 \begin{equation}
    \label{shifted_max_curve}
    \begin{aligned}
S(\angle h^*) &= \max_{i\in\{1,2,...,K\}} \  \hat{\beta}_i\cos(\angle h^*-\hat{\alpha}_i)\\
&= \max_{i\in\{1,2,...,K\}} \  \hat{\beta}_i\cos(\angle h^*-\hat{\alpha}_i - \angle v_n + \angle v_n)\\
&= \max_{i\in\{1,2,...,K\}} \  \hat{\beta}_i\cos(\angle h^*-\angle g_{n,i} + \angle v_n)\\
&=F_n(\angle h^* + \angle v_n).\\
\begin{comment}
F_n(\angle h^*) &= \max_i \  \hat{\beta}_i\cos(\angle h^*-\angle g_{n,i})\\
&= \max_i \  \hat{\beta}_i\cos(\angle h^*-\angle v_n-\hat{\alpha}_i)\\
&= \max_i \  \hat{\beta}_i\cos((\angle h^*-\angle v_n)-\hat{\alpha}_i)\\
&=S(\angle h^* - \angle v_n).
\end{comment}
    \end{aligned}
    \end{equation}
Equation (\ref{shifted_max_curve}) means that if we shift curve $F_n(\angle h^*)$, for any $n$, to the left by $\angle v_n$, then we can get curve $S(\angle h^*)$. From (\ref{shifted_max_curve}), we also have
\begin{equation}\label{shifted_max_curve_2}
F_n(\angle h^*) = S(\angle h^*- \angle v_n).
\end{equation}

Next, we introduce our method for configuration set selection.

\begin{comment}
As mentioned in ~\cref{theo1}, the configuration with the largest $\langle g_{n, i},h^*\rangle$ is the optimal configuration for the $n$th element. We later simplified $\langle g_{n, i},h^*\rangle$ to $\beta_i\cos(\angle h^*-\angle g_{n,i})$ in Section III-A. Therefore, if we were to select a number of 
\end{comment}
Consider $C_{\Psi}^*$ (the maximal capacity over a channel realization of the system with the configuration set $\Psi$). For presentation simplicity, we omit subscript `$\Psi$' and write $C_{\Psi}^*$ as $C^*$ in the sequel. For $C^*$, its expectation over channel realizations is expressed as
\begin{equation}
\label{probsimp}
\begin{aligned}
\mathbb{E}(C^*) = &\mathbb{E}(B\log_2(1+\frac{P|h^*|^2}{BN_0}))\\
= & B\ \mathbb{E}(\log_2(1+\frac{P|h^*|^2}{BN_0}))\\
= & B\lim_{R\rightarrow\infty}\frac{1}{R} \sum_{r=1}^R\log_2(1+\frac{P|h_r^*|^2}{BN_0})\\
\begin{comment}
= & \frac{B}{R}\ \log_2(\prod_{r=1}^R (1+\frac{P|h_r^*|^2}{BN_0}))\\
\approx & \frac{B}{R}\ \log_2(1+\frac{P}{BN_0}\sum_{r=1}^R|h_r^*|^2)\\
= & \frac{B}{R}\ \log_2(1+\frac{PR}{BN_0} \mathbb{E}(|h^*|^2))\\
\end{comment}
\end{aligned}
\end{equation}
in which $R$ (a very large number) is the number of channel realizations, and $h_r^*$ means the optimal overall channel from the transmitter to the receiver in the $r$th realization.
\begin{comment}
    
According to \eqref{probsimp}, since $log(x)$ is an increasing function, maximizing $\mathbb{E}(C)$ would be equivalent to maximizing $\mathbb{E}(|h^*|^2)$. 
\end{comment}

% Let us define $F_n(\angle h^*)$ and $S(\angle h^*)$ as 
% \begin{equation}
%    \label{max_curve_func}
%    \begin{aligned}
%&F_n(\angle h^*) = \max_i \  \hat{\beta}_i\cos(\angle h^*-\angle g_{n,i}),\\ %\quad \angle h^* \in I_{n,i}^*,\\ 
%&S(\angle h^*) = \max_i \  \hat{\beta}_i\cos(\angle h^*-\hat{\alpha}_i).
%    \end{aligned}
%    \end{equation}
%According to \eqref{concat_channel_ang}, we will have 

%Now if we shift the $\hat{\beta}_i\cos(\angle h^*-\angle g_{n,i})$ curves by $\angle v_n$, we will have the $\hat{\beta}_i\cos(\angle h^*-\alpha_{n,i})$. Let us define $S(\angle h^*)$ as the maximum of the $\hat{\beta}_i\cos(\angle h^*-\alpha_{n,i})$ curves. We will have
\begin{comment}
    
 \begin{equation}
    \label{shifted_max_curve}
    \begin{aligned}
&F_n(\angle h^*) = S(\angle h^* - \angle v_n).\\ 
    \end{aligned}
    \end{equation}
\end{comment}
%$F_n(\angle h^*)$ is depicted in Fig~\ref{Fn}.
\begin{figure}[!t]
\centering
\includegraphics[width=3.4in]{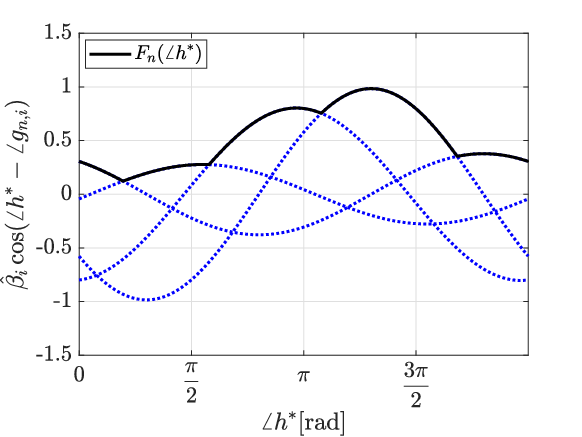}
\caption{An example demonstrating $F_n(\angle h^*)$}
\label{Fn}
\end{figure}

The term $|h_r^*|$ in (\ref{probsimp}) is expressed as:
\begin{equation}
    \label{max_curve_eq}
    \begin{aligned}
&|h_r^*| = \left|h_0+ \sum_{n=1}^N g_n^*\right| \\
&=\left|\frac{\langle h_0,h_r^*\rangle + \sum_{n=1}^N \langle g_{n}^*,h_r^*\rangle}{|h_r^*|}\right|\\
&=\left||h_0|\cos(\angle h_r^*-\angle h_0) + \sum_{n=1}^N |v_n|\beta_n^*\cos(\angle h_r^* - \angle g_n^*)\right|\\
&=\left||h_0|\cos(\angle h_r^*-\angle h_0) + \sum_{n=1}^N |v_n|F_n(\angle h_r^*)\right|.\\
%&=\left||h_0|\cos(\angle h_r^*) + c\sum_{n=1}^N F_n(\angle h_r^*)\right|\\
%&=\left||h_0|\cos(\angle h_r^*) + c\sum_{n=1}^N S(\angle h_r^* - \angle v_n)\right|\\
%&\approx\left||h_0|\cos(\angle h_r^*) + \frac{cN}{2\pi}\int_{0}^{2\pi} S(x)\,dx\right|\\
%&\approx\left|\frac{cN}{2\pi}\int_{0}^{2\pi} S(x)\,dx\right|.\\
    \end{aligned}
    \end{equation}
We assume $|v_1|, |v_2|,...,|v_N|$ are approximately the same and are equal to constant $c$.\footnote{This is a common assumption when RIS is at the far field of the transmitter and the receiver \cite{BjornsonWymeerschSPM2022}.} Also assuming a weak direct path ($|h_0| \approx 0$), $|h_r^*|$ can be further expressed from (\ref{max_curve_eq}) as
\begin{equation}
    \label{max_curve_eq_1}
    \begin{aligned}
|h_r^*| %= %\left|h_0+ \sum_{n=1}^N g_n^*\right| \\
%&=\left|\frac{\langle h_0,h_r^*\rangle + \sum_{n=1}^N \langle g_{n}^*,h_r^*\rangle}{|h_r^*|}\right|\\
%&=\left||h_0|\cos(\angle h_r^*) + \sum_{n=1}^N |v_n|\beta_n^*\cos(\angle h_r^* - \angle g_n^*)\right|\\
%=\left||h_0|\cos(\angle h_r^*) + \sum_{n=1}^N |v_n|F_n(\angle h_r^*)\right|\\
&\approx\left| c\sum_{n=1}^N F_n(\angle h_r^*)\right|\\
&\overset{\text{(i)}}=\left| c\sum_{n=1}^N S(\angle h_r^* - \angle v_n)\right|\\
&\overset{\text{(ii)}}\approx\left| \frac{cN}{2\pi}\int_{0}^{2\pi} S(x)\,dx\right|.
%&\overset{\text{(ii)}}\approx\left|\frac{cN}{2\pi}\int_{0}^{2\pi} S(x)\,dx\right|.\\
    \end{aligned}
    \end{equation}
Here step (i) is from (\ref{shifted_max_curve_2}), and step (ii) is to use the integral to replace Riemann sum.
As we can see in \eqref{max_curve_eq_1}, $|h_r^*|$ is proportional to $|\int_{0}^{2\pi} S(x)\,dx|$.
\begin{theorem}
\label{posareatheo}
    $\int_{0}^{2\pi} S(x)\,dx$ is always non negative.
\end{theorem}
\begin{proof}
From \eqref{S_func} we have
\begin{equation*}
    \label{positivearea}
    \begin{aligned}
    &S(\angle h^*) = \max_{i\in\{1,2,...,K\}} \  \hat{\beta}_i\cos(\angle h^*-\hat{\alpha}_i) \ge \hat{\beta}_1\cos(\angle h^*-\hat{\alpha}_1)\\
    &\longrightarrow\int_{0}^{2\pi} S(x)\,dx \ge \int_{0}^{2\pi} \hat{\beta}_1\cos(x-\hat{\alpha}_1)\,dx =0\\
    &\longrightarrow\int_{0}^{2\pi} S(x)\,dx \ge 0
    \end{aligned}
\end{equation*}
\begin{comment}
    
    Proof by contradiction can be used. Assume  $\int_{0}^{2\pi} S(x)\,dx < 0$. Thus, $\exists \ x_1 \in [0,2\pi) : S(x_1) < 0$. Assume $x_2 = (x_1 + \pi) \mod 2\pi$. We will have
\begin{equation}
    \label{positivearea}
    \begin{aligned}
    &S(x_2) \ge -S(x_1) \longrightarrow S(x_2) + S(x_1) \ge 0\\
    &\longrightarrow\int_{0}^{2\pi} S(x)\,dx \ge 0
    \end{aligned}
\end{equation}
which contradicts our initial assumption.
\end{comment}

\end{proof}

According to \eqref{probsimp}, \eqref{max_curve_eq_1} we have
\begin{equation}
    \label{finalprobform}
    \begin{aligned}
& \mathbb{E}(C^*) = B\lim_{R\rightarrow\infty}\frac{1}{R} \sum_{r=1}^R\log_2(1+\frac{P|h_r^*|^2}{BN_0})\\
=& B\lim_{R\rightarrow\infty}\frac{1}{R} \sum_{r=1}^R\log_2(1+\frac{P|\frac{cN}{2\pi}\int_{0}^{2\pi} S(x)\,dx|^2}{BN_0})\\
=& B\lim_{R\rightarrow\infty}\frac{1}{R} \sum_{r=1}^R\log_2(1+\frac{Pc^2N^2|\int_{0}^{2\pi} S(x)\,dx|^2}{4\pi^2BN_0})\\
=& B\lim_{R\rightarrow\infty}\frac{1}{R} R\log_2(1+\frac{Pc^2N^2|\int_{0}^{2\pi} S(x)\,dx|^2}{4\pi^2BN_0})\\
=& B\log_2(1+\frac{Pc^2N^2|\int_{0}^{2\pi} S(x)\,dx|^2}{4\pi^2BN_0})\\
=& B\log_2(1+\frac{Pc^2N^2(\int_{0}^{2\pi} S(x)\,dx)^2}{4\pi^2BN_0}),
    \end{aligned}
\end{equation}
in which the last equality comes from the fact that $\int_{0}^{2\pi} S(x)\,dx$ is always non-negative (Theorem \ref{posareatheo}).
According to \eqref{finalprobform}, since $\log(x)$ is an increasing function, maximizing $\mathbb{E}(C^*)$ would be equivalent to maximizing $\int_{0}^{2\pi} S(x)\,dx$, which is a major insight of our method.
%Moreover, according to ~\cref{posareatheo}, since $\int_{0}^{2\pi} S(x)\,dx$ is always non negative, 
%$|\int_{0}^{2\pi} S(x)\,dx|^2$ will become maximum when  $\int_{0}^{2\pi} S(x)\,dx$ is maximized . 

%Thus, maximizing $\mathbb{E}(C^*)$ will be equivalent to maximizing $\int_{0}^{2\pi} S(x)\,dx$. 
So in our method,  when we go through all ${M}\choose{K}$ options of $\Psi$, we no longer need to consider Monte Carlo simulations with a large number of channel realizations. For each option, we only need to compute $\int_{0}^{2\pi} S(x)\,dx$. The $\Psi$ option corresponding to the largest $\int_{0}^{2\pi} S(x)\,dx$ will be considered as the optimal configuration set. Since our method finds the maximal integral of $S(x)$, we call our method {\it Integral Maximization Based (IMB) method}.

\begin{figure}[!t]
\centering
\includegraphics[width=3.4in]{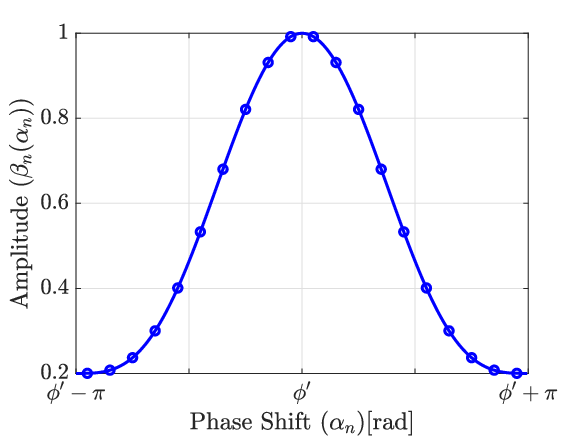}
\caption{Curve $\beta_n(\alpha_n)$ vs. $\alpha_n$ over phase shift range $[\phi' - \pi, \phi' + \pi)$, with an example to choose $M=20$ points on the curve.}
\label{amp_phas_confsel}
\end{figure}
\subsection{Search Space Compression (SSC)}
In Section \ref{sec:IMB}, given $\Omega$ (a set of $M$ evenly distributed phase shifts over range $[0,2\pi$)), the proposed IMB method selects the best option of $\Psi$ among all ${M}\choose{K}$ options. Next, we show how we should pick up the $M$ evenly distributed phase shifts over range $[0,2\pi)$.

Picking up $M$ evenly distributed phase shifts over range $[0,2\pi)$ is actually picking up $M$ points on the curve in Fig.~\ref{amp_phas_relation} over phase shift range $[0,2\pi)$. Note that in Fig.~\ref{amp_phas_relation}, the curve $\beta_n(\alpha_n)$ is symmetric over line $\alpha_n = \phi'$ (in which $\phi'=\phi + \frac{\pi}{2}$ with $\phi$ being a constant) if we view the range of $\alpha_n$ as $(-\infty, \infty)$. Thus, we select to pick up $M$ points on the curve over phase shift range $[\phi' - \pi, \phi' +\pi)$, as shown in Fig.~\ref{amp_phas_confsel}.\footnote{Picking up $M$ points over phase shift range $[0,2\pi)$ is equivalent to picking up $M$ points over phase shift range $[\phi' - \pi, \phi' +\pi)$.} The curve in Fig.~\ref{amp_phas_confsel} is perfectly symmetric over line $\alpha_n = \phi'$.

To pick up $M$ points on the symmetric curve in Fig.~\ref{amp_phas_confsel}, intuitively the $M$ points should be symmetric over line  $\alpha_n = \phi'$. Accordingly, $\Omega$ should be expressed as
\begin{equation}
    \label{eq:omega_exp}
    \begin{aligned}
\Omega = \{ \phi^\prime - \pi + \frac{\pi}{M},\phi^\prime - \pi + \frac{3\pi}{M},...,\phi^\prime - \pi + \frac{(2M-1)\pi}{M}\}.
    \end{aligned}
\end{equation}
Fig.~\ref{amp_phas_confsel} also shows an example of how $M=20$ points can be chosen on the curve.

Recall that for a given $\Omega$, our proposed IMB method should go through all ${M}\choose{K}$ options of $\Psi$. Next, we show that for $\Omega$ given in (\ref{eq:omega_exp}), some options of $\Psi$ yield the same $\int_{0}^{2\pi} S(x)\,dx$, and thus, we actually do not have to go through all ${M}\choose{K}$ options.

%According to \eqref{amp_phas_relation}, $\beta_n(\alpha_n)$ is a symmetric function. Because of this symmetry, we expect that some options of $\Psi$ might yield the same $\int_{0}^{2\pi} S(x)\,dx$. Thus, it would be unnecessary to go through all of them. In this section, we will discuss how these cases can be identified.

Consider an option $\Psi$ from $\Phi$ (recalling that $\Phi$ is the set of all $K$-size subsets of $\Omega$). In $\Psi$, we have $K$ phase shifts, which are corresponding to $K$ points (among the $M$ points) on the curve in Fig.~\ref{amp_phas_confsel}. For presentation simplicity, we denote $\Psi$ as $\Psi=\{\psi_{1}, \psi_{2}, ..., \psi_{K}\}$, in which $\psi_{1}, \psi_{2}, ..., \psi_{K}$ are the $K$ reflection coefficient choices of option $\Psi$ (which are also $K$ points among the $M$ points on the curve in Fig.~\ref{amp_phas_confsel}). Now consider another option from $\Phi$, denoted $\Psi^\dag=\{\psi^\dag_{1}, \psi^\dag_{2}, ..., \psi^\dag_{K}\}$, in which point $\psi^\dag_{k}$ and point $\psi_k$ ($k=1,2,...,K$) are symmetric over the symmetric line $\alpha_n=\phi'$ in Fig.~\ref{amp_phas_confsel}. In other words, $\psi^\dag_{k}$ and $\psi_k$ have the same amplitude but their phases are mirrored over the symmetric line, i.e., $\frac{\angle \psi^\dag_k + \angle \psi_k}{2} = \phi + \frac{\pi}{2}$. We say option $\Psi^\dag$ and option $\Psi$ are {\it mirrored option} to each other. For the two options, the $S(x)$ function is denoted as $S_\Psi(x)$ and $S_{\Psi^\dag}(x)$, respectively. We have
\begin{equation}
    \label{spacereduc}
    \begin{aligned}
    S_{\Psi^\dag}(x) &= \max_{k=1,2,...,K} |\psi^\dag_k| \cos(x - \angle \psi^\dag_k)\\
&=\max_{k=1,2,...,K} |\psi_k| \cos(x -(2\phi + \pi  - \psi_k) )\\ 
%&=\max_k |\psi_k| \cos((2\phi + \pi  - \psi_k) -\angle h^* )\\ 
&=\max_{k=1,2,...,K} |\psi_k| \cos((2\phi + \pi-x)  - \psi_k  )\\ 
&= S_\Psi(2\phi + \pi-x).
    \end{aligned}
\end{equation}
Then we have
\begin{equation}
    \label{spacereduc_integ}
    \begin{aligned}
    \int_{0}^{2\pi}  S_{\Psi^\dag}(x)\,dx &= \int_{0}^{2\pi} S_\Psi(2\phi + \pi-x)\,dx\\
    &\overset{\text{(iii)}}=\int_{-2\pi}^{0} S_\Psi(2\phi + \pi+x)\,dx\\
    &\overset{\text{(iv)}}=\int_{2\phi-\pi}^{2\phi+\pi} S_\Psi(x)\,dx\\
    &\overset{\text{(v)}}=\int_{0}^{2\pi} S_\Psi(x)\,dx.\\
    \end{aligned}
\end{equation}
Here in step (iii) we replace $-x$ by $x$, in step (iv) we replace $2\phi + \pi+x$ with $x$, and in step (v) we use the fact that $S_\Psi(x)$ is a periodical function with period $2\pi$.

Equation (\ref{spacereduc_integ}) shows that for the two options $\Psi$ and $\Psi^\dag$, the integral of $S_{\Psi}(x)$ and $S_{\Psi^\dag}(x)$ are the same. Thus, we only need to check one of the two options. We call this as {\it Search Space Compression (SSC)}.
\begin{itemize}
    \item If $K$ is an even number, then among all ${M}\choose{K}$ options of $\Psi$, some options are identical to their mirrored options, and the number of such options is ${\lfloor\frac{M}{2}\rfloor}\choose{\frac{K}{2}}$. Thus, the total number of options of $\Psi$ that need to be checked is ${{\lfloor\frac{M}{2}\rfloor}\choose{\frac{K}{2}}} + \frac{{{M}\choose{K}} - {{\lfloor\frac{M}{2}\rfloor}\choose{\frac{K}{2}}}}{2}$, which is approximately ${{{M}\choose{K}}}/{2}$ since $M$ is large. 

    \item If $K$ is an odd number, there will be two cases. When $M$ is even, each option is different from its mirrored option, and thus, the total number of options of $\Psi$ that need to be checked is ${{{M}\choose{K}}}/{2}$. When $M$ is odd, the number of options of $\Psi$ that are identical to their mirrors is  ${{\frac{M-1}{2}}\choose{\frac{K-1}{2}}}$. Thus, the total number of options of $\Psi$ that need to be checked will be ${{\frac{M-1}{2}}\choose{\frac{K-1}{2}}} + \frac{{{M}\choose{K}} - {{\frac{M-1}{2}}\choose{\frac{K-1}{2}}}}{2}$, which is approximately ${{{M}\choose{K}}}/{2}$ since $M$ is large.
    
\end{itemize}

Therefore, by the SSC method, the number of options that should be checked is cut approximately by half.

\section{Numerical Results}\label{sec:sim}
In this section, we will evaluate the performance of our proposed methods for capacity maximization and configuration set selection.
In the following simulations, the parameters are set according to Table~\ref{simul_table} unless specified otherwise.
\subsection{Capacity Maximization}

We simulate our capacity maximization method in Section \ref{sec:cap_max_section}  as well as three benchmark methods as follows. 
\begin{itemize}
\item Exhaustive search method: we go through all $K^N$ possibilities of \{$\theta_1,\theta_2,...,\theta_N$\} to find the optimal phases. 

\item Closest point projection (CPP): CPP is a heuristic algorithm used in \cite{BjornsonWymeerschSPM2022}. The idea of this algorithm is to align all RIS channels toward the direct channel as much as possible. In other words, $g_n^*$ would be the one that maximizes $\cos(\angle h_0 - \angle g_{n, i})$. 

\item Improved CPP: In the CPP method in \cite{BjornsonWymeerschSPM2022}, phase shift and amplitude of a reflection coefficient can be independently adjusted. Since we consider $\beta_n$ and $\alpha_n$ to be coupled, we make some changes to the original CPP method by using our result in Theorem 1 as follows. Instead of aligning all RIS channels toward the direct channel, we maximize the inner product of each RIS channel with the direct channel. This means we are maximizing the projection of all RIS channels on the direct channel. Thus, $g_n^*$ will now be the one that maximizes $\hat{\beta}_i\cos(\angle h_0 - \angle g_{n, i})$. This method is called {\it improved CPP}.

\end{itemize}

\begin{table}[!t]
\caption{Parameter values for simulation results \label{simul_table}}
\centering
\begin{tabular}{|c|c|c|c|}
\hline
Parameter & Value & Parameter & Value\\
\hline
$B$&$1$ \ MHz &$\frac{P}{BN_0}$&$100$ \ dB\\
\hline
$\beta_{\min}$&$0.2$ & $\phi$&$0.43\pi$\\
\hline
$\kappa$&$1.6$ & $M$&$20$\\
\hline
$|v_n|$&$-140$ \ dB & $\angle v_n$&$\sim$ Uniform$[0,2\pi)$\\
\hline
$|h_0|$&$-140$ \ dB & $\angle h_0$&$0$\\
\hline

\end{tabular}
\end{table}

%comment only for conf
\begin{comment}
\begin{table}[!t]
\caption{Parameter values for simulation results \label{simul_table}}
\centering
\begin{tabular}{|c|c|}
\hline
Parameter & Value\\
\hline
$R$&$1000$\\
\hline
$B$&$1 \ MHz$\\
\hline
$\frac{P}{BN_0}$&$100 \ dB$\\
\hline
$\beta_{min}$&$0.2$\\
\hline
$\phi$&$0.43\pi$\\
\hline
$\kappa$&$1.6$\\
\hline
$\alpha_n$&$\{0, \frac{2\pi}{K}, ..., \frac{2\pi(K-1)}{K} \}$\\
\hline
$|v_n|$&$-140 \ dB$\\
\hline
$\angle v_n$&$U(0,2\pi)$\\
\hline
$|h_0|$&$-140 \ dB$\\
\hline
$\angle h_0$&$0$\\
\hline
$M$&$36$\\
\hline

\end{tabular}
\end{table}
\end{comment}
In our simulations, the reflection coefficient of each RIS element is chosen from $K$ choices as shown in (\ref{configuration_set}), while the $K$ choices have evenly distributed phase shifts, i.e., $\hat{\alpha}_k = \frac{(k-1)\times 2\pi}{K}, k=1,2,...,K$.

Fig.~\ref{capvsnk2} shows how capacity changes with the number of elements for different algorithms with $K=2$. According to \eqref{Total_channel} and \eqref{capacity_eq}, we expect the capacity to be an increasing function of $N$, which is verified by the four curves in Fig.~\ref{capvsnk2}. 
As seen in Fig.~\ref{capvsnk2}, for all values of $N$, our proposed method yields the same capacity as the exhaustive search method, which means that our method can achieve optimality with linear complexity. The original CPP and the improved CPP have the same performance. This happens because $K$ is set to two, hence $\hat{\alpha}_{1}$ and $\hat{\alpha}_{2}$ are $\pi$ radians apart. Thus, $\cos(\angle h_0 - \angle g_{n, 1})$ and $\cos(\angle h_0 - \angle g_{n, 2})$ will have opposite signs. As a result, the amplitude no longer matters.

\begin{figure}[!t]
\centering
\includegraphics[width=3.4in]{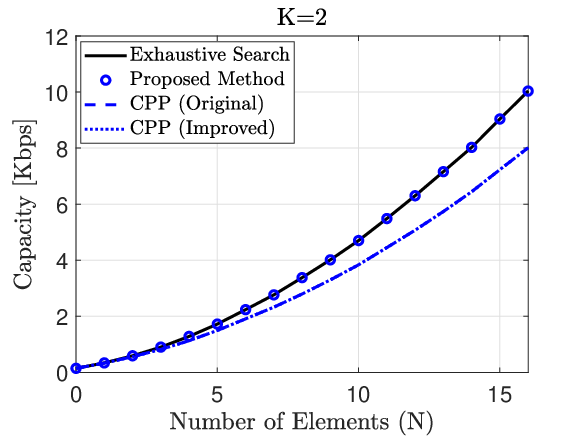}
\caption{Capacity versus the number of elements (with $K=2$ choices of reflection coefficients).}
\label{capvsnk2}
\end{figure}

In Fig.~\ref{capvsnk4}, $K$ is set to 4. Our method and the exhaustive search method still have the optimal performance. By using our result in Theorem 1, improved CPP outperforms the original CPP but still yields a suboptimal solution.

\begin{figure}[!t]
\centering
\includegraphics[width=3.4in]{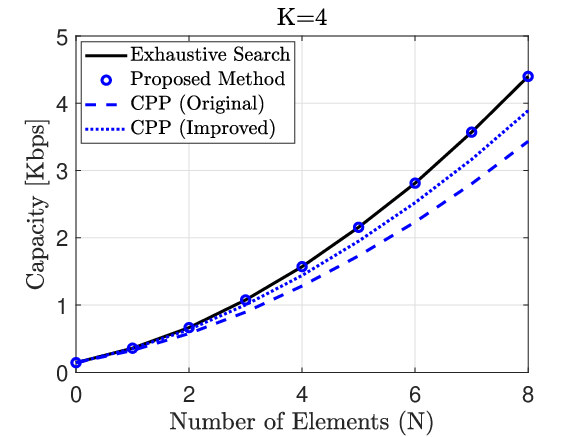}
\caption{Capacity versus the number of elements (with $K=4$ choices of reflection coefficients).}
\label{capvsnk4}
\end{figure}

In Fig.~\ref{capvsh0}, we examine how the strength of the direct channel affects the performance of the mentioned methods. According to \eqref{Total_channel}, \eqref{capacity_eq}, the capacity is expected to be an increasing function of $|h_0|$. At $|h_0|=-140~\text{dB}$, the proposed method has a noticeable advantage over the improved CPP. But as $|h_0|$ increases, the gap between the two diminishes. The reason is that when the direct channel becomes stronger, $h_0$ will become the dominant term in \eqref{Total_channel}, and thus, its phase and amplitude greatly affect $h^*$. As a result, $\angle h_0$ will become an appropriate approximation for $\angle h^*$. In other words, the improved CPP would be a proper estimate for our proposed method when the direct channel is strong.

\begin{figure}[!t]
\centering
\includegraphics[width=3.4in]{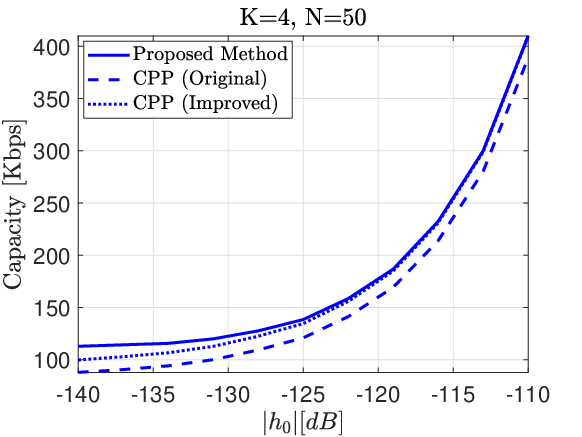}
\caption{Capacity versus $|h_0|$.}
\label{capvsh0}
\end{figure}

\subsection{Configuration Set Selection}
Next, we use simulations to evaluate the performance of our configuration set selection method IMB as well as the IMB method enhanced with the SSC method (denoted as ``IMB+SSC"). As a comparison, we also simulate two other methods: 1) the MCSB method with $R=1,000$ channel realizations for each simulation setup, and 2) the evenly distributed configuration set selection method in which the phase shifts of the reflection coefficients in the configuration set are evenly distributed (i.e., $\{\hat{\alpha}_1, \hat{\alpha}_2,...,\hat{\alpha}_K\}=\{0, \frac{2\pi}{K},\frac{4\pi}{K} ..., \frac{2\pi(K-1)}{K}\}$). 

Fig.~\ref{capvsnk2conf} demonstrates the performance of different methods used for configuration set selection. As the MCSB method uses Monte Carlo Simulations, it can be viewed as the optimal method.  As we can see in Fig.~\ref{capvsnk2conf}, MCSB achieves the maximal capacity, while our IMB and IMB+SSC have the same performance with almost negligible difference from the performance of MCSB, which means that our IMB and IMB+SSC achieve an almost-optimal performance, and the SSC method reduces search space without any performance degradation. The evenly distributed configuration set selection method has less capacity than MCSB, IMB, and IMB+SSC.

Since IMB and IMB+SSC have the same capacity performance, we do not show the results of IMB+SSC in Figs.~\ref{capvsk}-\ref{capvskappa}.

Fig.~\ref{capvsk} depicts how capacity changes with $K$ in our IMB method and the evenly distributed configuration set selection method. The performance of MCSB method is not shown in this figure, due to the prohibitive simulation time needed for the MCSB method. As we expected, in both IMB and the evenly distributed configuration set selection methods, increasing $K$ would provide us with a capacity gain. The gain is large for small values of $K$ (e.g. from $K=2$ to $K=4$). This suggests that increasing $K$ to a large number would be unnecessary and considering small values for $K$ (e.g., $K=8$) would be sufficient. %Regarding the performance gap between the two curves, we can see IMB outperforms the evenly distributed method with a larger gap for smaller values of $K$. 
\begin{comment}
The capacity ratio between IMB and CMB is depicted as well. As $N$ grows, we can see that the ratio increases. Let us discuss the reason. In \eqref{max_curve_eq}, we replaced the summation with the integral using the Riemann sum approximation. The accuracy of this approximation increases with $N$. So we expect IMB to be more accurate as $N$ gets larger which is why the ratio increases with $N$.
\end{comment}
\begin{figure}[!t]
\centering
\includegraphics[width=3.4in]{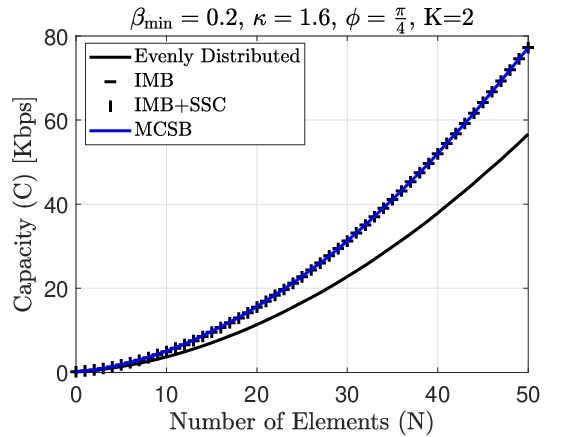}
\caption{Capacity versus the number of elements for different configuration set selection methods.}
\label{capvsnk2conf}
\end{figure}

\begin{figure}[!t]
\centering
\includegraphics[width=3.4in]{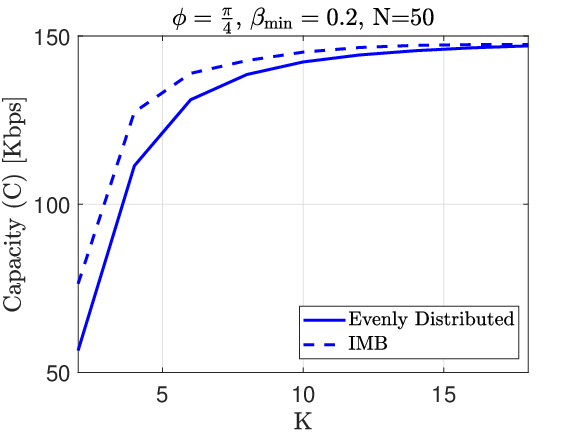}
\caption{Capacity versus the number of choices of reflection coefficients $(K)$}
\label{capvsk}
\end{figure}
Next, we will discuss how the parameters in the reflection coefficient model in \eqref{couple_formula} influences the performance of different configuration set selection methods. 

\begin{comment}
    
As we can see in Fig.~\ref{capvsphi}, the performance of the evenly distributed method is highly dependent on the value of $\phi$. In contrast, the performance of our proposed methods remains consistent regardless of the changes in $\phi$. Changing the value of $\phi$ is equivalent to shifting the $\beta_n(\alpha_n)$ function. Thus, our proposed methods go through the same choices but in a different order. As a result, their performances remain unchanged.
\end{comment}

\begin{comment}
    
\begin{figure}[!t]
\centering
\includegraphics[width=3.4in]{capvsphi.png}
\caption{Capacity versus $\phi$}
\label{capvsphi}
\end{figure}
\end{comment}

Fig.~\ref{capvsbetamin} shows how $\beta_{\min}$ affects the capacity of the configuration set selection methods. According to \eqref{couple_formula}, $\beta_{\min}$ represents the amount of loss in an RIS element. High $\beta_{\min}$ indicates that the element has low loss whereas low $\beta_{\min}$ implies that the element is quite lossy. Thus, we expect the capacity to be an increasing function of $\beta_{\min}$ for all methods. When $\beta_{\min} = 1$, equation \eqref{couple_formula} reduces to $\beta_n(\alpha_n) = 1$, which means that the amplitude $\beta_n$ and phase shift $\alpha_n$ are not coupled anymore, and thus, any set of $K$ reflection coefficients whose phase shifts are evenly spaced would be the optimal solution. So all methods yield the same capacity at $\beta_{\min} = 1$. It can also be observed that our proposed method (IMB) is most effective when the RIS elements are highly lossy, i.e., when $\beta_{\min}$ is small. 
\begin{figure}[!t]
\centering
\includegraphics[width=3.4in]{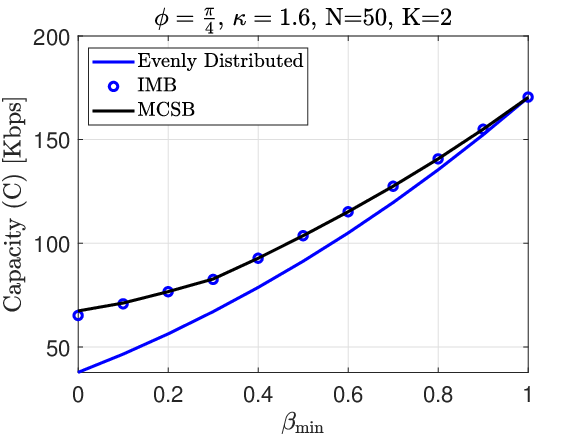}
\caption{Capacity versus $\beta_{\min}$.}
\label{capvsbetamin}
\end{figure}

We can see the effect of $\kappa$ of \eqref{couple_formula} on the performances of the configuration set selection methods in Fig.~\ref{capvskappa}. Similar to $\beta_{\min}$, $\kappa$ is also an indicator of the degree of loss in an RIS element. In contrast to $\beta_{\min}$, the value $\kappa$ is proportional to the amount of loss. As a result, we expect the achievable capacity to be a decreasing function of $\kappa$. Similar to Fig.~\ref{capvsbetamin}, in the lossless scenario ($\kappa=0$), we have $\beta_n(\alpha_n) = 1$, and thus, all methods achieve the same capacity. %As $\kappa$ increases, we can see the performance gap between our proposed method and the evenly distributed method gets larger. %This suggests that our proposed method is the best fit for cases where the degree of loss is high.
\begin{figure}[!t]
\centering
\includegraphics[width=3.4in]{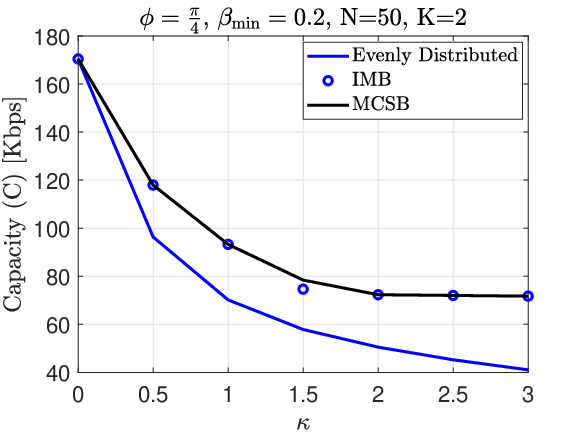}
\caption{Capacity versus $\kappa$.}
\label{capvskappa}
\end{figure}

Next, we demonstrate the benefit of IMB+SSC compared to IMB. Fig.~\ref{proctimevsm} demonstrates the processing time for IMB and IMB+SSC. The processing time is defined as the time that a method takes during determining the configuration set. As we can see, the IMB+SSC method is almost twice as fast as the original IMB. Fig.~\ref{numstatesvsm} shows the number of options of $\Psi$ that each method has to go through. As we expected, when SSC is applied to IMB, the number of searched options gets almost halved resulting in a more compact search space. 

\begin{figure}[!t]
\centering
\includegraphics[width=3.4in]{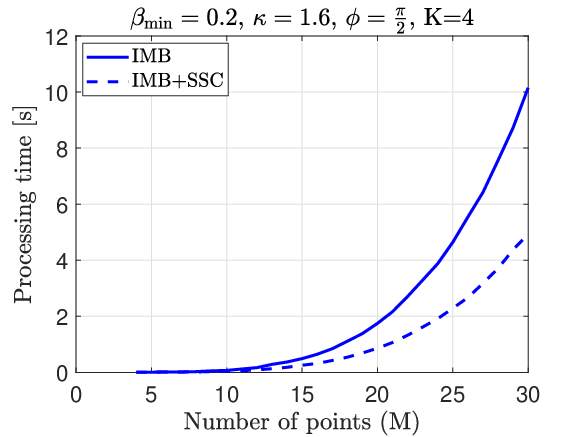}
\caption{Processing time versus $M$.}
\label{proctimevsm}
\end{figure}

\begin{figure}[!t]
\centering
\includegraphics[width=3.5in]{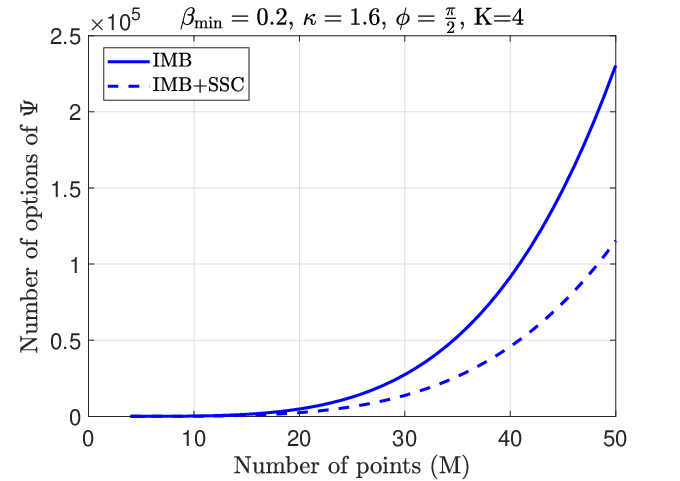}
\caption{Number of searched options of $\Psi$ versus $M$.}
\label{numstatesvsm}
\end{figure}

\section{Conclusion}\label{sec:conclusion}
%In this paper, we work on the discrete reflection optimization of RIS under a practical model for the amplitude and phase shift of reflection coefficients. To maximize the channel capacity, we first find the optimal reflection coefficients of the RIS elements by assuming $\angle h^*$ is known. Then we show that to get the overall optimal solution, we do not need to go through the infinite possibilities of $\angle h^*$. Rather, we only need to check a finite number of $\angle h^*$ regions, and then get the overall optimal reflection coefficients of the RIS elements. Our method has a complexity liner with the number of RIS elements and the number of reflection coefficient choices. Although our work is focused on a system with a single user, our work can be straightforwardly extended to multiple-user cases, as the RIS partitioning technology \cite{MakinArzykulovTWC2024,KhaleelBasarJSTSP2022} and/or distributed RIS deployment \cite{ZhangZhangTCOM2021} can decompose a multi-user scenario into a number of single-user sub-problems.
In this paper, we work on the discrete reflection optimization of RIS elements. In contrast to most works in the literature, we consider a practical setup in which the amplitude and the phase shift of RIS elements are coupled. To maximize the capacity of a system with a given configuration set, we develop an algorithm that yields the global optimal reflection coefficients of RIS elements with linear complexity. We also develop an efficient method called ``IMB'' that finds the optimal configuration set. Our method is based on our insightful finding that maximizing the average system capacity is approximately equivalent to maximizing the integral $\int_{0}^{2\pi} S(x)\,dx$. %Thus, our method is much faster than the Monte Carlo simulation based method that has to use a large number of simulations for varying channel realizations. 
Numerical results show that our capacity maximization method and configuration selection method have apparent gains in terms of channel capacity. In this work, we investigate a single-user setup. However, as discussed in Section I, techniques such as RIS partitioning and/or distributed RIS deployment can help us straightforwardly extend our methods to multi-user setups.

\appendices

\section{Proof of Theorem 1}
    We use proof by contradiction. Assume $g_n^* \neq g_{n,i}$. Let us consider $g_n^*=g_{n,l}, (l\in\{1,2,...,K\},l\neq i)$. From the optimal reflection coefficients of all RIS elements that achieve $h^*$, if we replace the reflection coefficient of the $n$th RIS element with $g_{n,i}$, then the overall channel from the transmitter to the receiver is denoted as $h^\dag=h^* - g_{n,l} + g_{n,i}$. Since we assume that $g_{n,i}$ is not optimal, $|h^\dag|$ should be smaller than $|h^*|$. We will have:
    \begin{equation}
    \label{theo1_eq}
    \begin{aligned}
    &|h^\dag|^2 < |h^*|^2\\
    \longrightarrow \quad& |h^* - g_{n,l} + g_{n,i}|^2 < |h^*|^2 \\
    \longrightarrow \quad& |h^* - g_{n,l}|^2 + |g_{n,i}|^2 + 2\langle h^* - g_{n,l},g_{n,i}\rangle < |h^*|^2 \\  
    \overset{\text{(i)}}{\longrightarrow} \quad& |h^* - g_{n,l}|^2 + |g_{n,i}|^2 + 2\langle h^*,g_{n,i}\rangle  \\  
    &-2\langle g_{n,l},g_{n,i}\rangle < |h^*|^2\\
    &\text{(step (i) is due to additivity property of inner product)}\\
    \longrightarrow \quad& |h^*|^2 + |g_{n,l}|^2 - 2\langle h^*,g_{n,l}\rangle + |g_{n,i}|^2 + 2\langle h^*,g_{n,i}\rangle  \\ 
    &-2\langle g_{n,l},g_{n,i}\rangle < |h^*|^2\\
    \longrightarrow \quad& |g_{n,l}|^2 + |g_{n,i}|^2 -2\langle g_{n,l},g_{n,i}\rangle - 2\langle h^*,g_{n,l}\rangle    \\ 
    & + 2\langle h^*,g_{n,i}\rangle < 0\\  
    \longrightarrow \quad& |g_{n,l} - g_{n,i}|^2 - 2\langle h^*,g_{n,l}\rangle + 2\langle h^*,g_{n,i}\rangle < 0\\ 
    \longrightarrow \quad& |g_{n,l} - g_{n,i}|^2 + 2(\langle h^*,g_{n,i}\rangle - \langle h^*,g_{n,l}\rangle) < 0.\\ 
    \end{aligned}
    \end{equation}
    $|g_{n,l} - g_{n,i}|^2$ is a non-negative number. Since $\langle h^*,g_{n,i}\rangle$ is the maximum among $\{\langle h^*,g_{n,1}\rangle,\langle h^*,g_{n,2}\rangle,...,\langle h^*,g_{n,K}\rangle\}$, $2(\langle h^*,g_{n,i}\rangle - \langle h^*,g_{n,l}\rangle)$ is also a non-negative number. Thus, we have reached a contradiction in the last line of (\ref{theo1_eq}). The proof is now complete.

\section{Proof of Theorem 2}
Consider curve $\hat{\beta}_i\cos(\angle h^*-\angle g_{n,i})$. Assume the interval between two consecutive intersections $q_1$ and $q_2$ on the curve\footnote{Here $q_1$ and $q_2$ are $\angle h^*$ values of the two intersections.} is an active interval. Assume $q_2$ is the intersection in common between curve $\hat{\beta}_i\cos(\angle h^*-\angle g_{n,i})$ and curve $\hat{\beta}_l\cos(\angle h^*-\angle g_{n,l})$. The interval from $q_1$ to $q_2$ is assumed to be active for curve $\hat{\beta}_i\cos(\angle h^*-\angle g_{n,i})$. Thus, we have $\hat{\beta}_i\cos(\angle h^*-\angle g_{n,i}) > \hat{\beta}_l\cos(\angle h^*-\angle g_{n,l}), \forall \angle h^* \in (q_1, q_2)$. At the beginning of Section III-C, we have proved that the intersections in common between each pair of curves are $\pi$ radians apart. Therefore, we can say $\hat{\beta}_i\cos(\angle h^*-\angle g_{n,i}) > \hat{\beta}_l\cos(\angle h^*-\angle g_{n,l}), \forall \angle h^* \in (q_2 -\pi, q_2)$.\footnote{Recall that each interval is defined within $[0,2\pi)$ in Section \ref{sec:inf_to_finite}. Here interval $(q_2 -\pi, q_2)$ actually means $(q_2 -\pi \mod 2\pi, ~q_2)$. We use $(q_2 -\pi, q_2)$ for presentation simplicity. In general, when we write an interval as $(x_1,x_2)$, it actually means $(x_1 \mod 2\pi,~x_2 \mod 2\pi)$.} We will have:
     \begin{equation}
    \label{q_2_interv}
    \begin{aligned}
    &\hat{\beta}_i\cos(\angle h^*-\angle g_{n,i}) > \hat{\beta}_l\cos(\angle h^*-\angle g_{n,l}),\\ &\forall \angle h^* \in (q_2 -\pi, q_2)\\
    &\longrightarrow - \hat{\beta}_i\cos(\angle h^*-\angle g_{n,i}) <- \hat{\beta}_l\cos(\angle h^*-\angle g_{n,l}),\\ &\forall \angle h^* \in (q_2 -\pi, q_2)\\
    &\rightarrow \hat{\beta}_i\cos(\angle h^*+\pi-\angle g_{n,i}) < \hat{\beta}_l\cos(\angle h^*+\pi-\angle g_{n,l}),\\ &\forall \angle h^* \in (q_2 -\pi, q_2)\\
    &\rightarrow \hat{\beta}_i\cos(\angle h^*-\angle g_{n,i}) < \hat{\beta}_l\cos(\angle h^*-\angle g_{n,l}),\\ &\forall \angle h^* \in (q_2, q_2+\pi).\\
    \end{aligned}
    \end{equation}
According to \eqref{q_2_interv}, $\hat{\beta}_i\cos(\angle h^*-\angle g_{n,i})$ cannot be the maximum curve (i.e., the curve above all other curves) $\forall \angle h^* \in (q_2, q_2+\pi)$. Now, assume $q_1$ is the intersection in common between $\hat{\beta}_i\cos(\angle h^*-\angle g_{n,i})$ and $\hat{\beta}_{l^\prime}\cos(\angle h^*-\angle g_{n,l^\prime})$. We have $\hat{\beta}_i\cos(\angle h^*-\angle g_{n,i}) > \hat{\beta}_{l^\prime}\cos(\angle h^*-\angle g_{n,l^\prime}), \forall \angle h^* \in (q_1, q_2)$. 
Since the interval from $q_1$ to $q_2$ is active for curve $\hat{\beta}_i\cos(\angle h^*-\angle g_{n,i})$, we have $\hat{\beta}_i\cos(\angle h^*-\angle g_{n,i}) > \hat{\beta}_{l^\prime}\cos(\angle h^*-\angle g_{n,l^\prime}), \forall \angle h^* \in (q_1, q_1 + \pi)$. We will have:
     \begin{equation}
    \label{q_1_interv}
    \begin{aligned}
    &\hat{\beta}_i\cos(\angle h^*-\angle g_{n,i}) > \hat{\beta}_{l^\prime}\cos(\angle h^*-\angle g_{n,l^\prime}),\\ &\forall \angle h^* \in (q_1, q_1 + \pi)\\
    &\longrightarrow - \hat{\beta}_i\cos(\angle h^*-\angle g_{n,i}) <- \hat{\beta}_{l^\prime}\cos(\angle h^*-\angle g_{n,l^\prime}),\\ &\forall \angle h^* \in (q_1, q_1 + \pi)\\
    &\rightarrow \hat{\beta}_i\cos(\angle h^*-\pi-\angle g_{n,i}) < \hat{\beta}_{l^\prime}\cos(\angle h^*-\pi-\angle g_{n,l^\prime}),\\ &\forall \angle h^* \in (q_1, q_1 + \pi)\\
    &\rightarrow \hat{\beta}_i\cos(\angle h^*-\angle g_{n,i}) < \hat{\beta}_{l^\prime}\cos(\angle h^*-\angle g_{n,l^\prime}),\\ &\forall \angle h^* \in (q_1 - \pi, q_1).\\
    \end{aligned}
    \end{equation}
According to \eqref{q_1_interv}, $\hat{\beta}_i\cos(\angle h^*-\angle g_{n,i})$ cannot be the maximum curve $\forall \angle h^* \in (q_1 - \pi, q_1)$. Since $(q_1 - \pi, q_1) \cup (q_1, q_2) \cup (q_2, q_2 + \pi) $ covers the whole $[0,2\pi)$ range, there will be no active interval outside $(q_1, q_2)$ for curve $\hat{\beta}_i\cos(\angle h^*-\angle g_{n,i})$. This completes the proof.

\end{document}